\documentclass[%
reprint,
superscriptaddress,
amsmath,amssymb,
prl,
]{revtex4-2}

\usepackage[unicode=true,colorlinks=true,citecolor=blue,urlcolor=blue]{hyperref}

\usepackage{gensymb}
\usepackage{graphicx}
\usepackage{bm}

\usepackage{amsthm}
\usepackage{physics}

\usepackage[addedmarkup=sl,xcolor=dvipsnames,defaultcolor = magenta]{changes}

\theoremstyle{definition}
 
\theoremstyle{theorem}

\usepackage{mathtools}

\makeatletter
\renewcommand*\env@matrix[1][\arraystretch]{%
  \edef\arraystretch{#1}%
  \hskip -\arraycolsep
  \let\@ifnextchar\new@ifnextchar
  \array{*\c@MaxMatrixCols c}}
\makeatother

\begin{document}

\title{
Experimental realization of stable exceptional chains protected by non-Hermitian latent symmetries unique to mechanical systems
}

\author{Xiaohan Cui}
\thanks{These authors contributed equally to this work.}
\affiliation{%
Department of Physics, The Hong Kong University of Science and Technology, Hong Kong, China
}%

\author{Ruo-Yang Zhang}
\thanks{These authors contributed equally to this work.}
\affiliation{%
Department of Physics, The Hong Kong University of Science and Technology, Hong Kong, China
}

\author{Xulong Wang}
\affiliation{%
Department of Physics, Hong Kong Baptist University, Kowloon Tong, Hong Kong, China
}%
\author{Wei Wang}
\affiliation{%
Department of Physics, Hong Kong Baptist University, Kowloon Tong, Hong Kong, China
}
\author{Guancong Ma}%
\email{phgcma@hkbu.edu.hk}
\affiliation{%
Department of Physics, Hong Kong Baptist University, Kowloon Tong, Hong Kong, China
}%
\author{C. T. Chan}%
\email{phchan@ust.hk}
\affiliation{%
Department of Physics, The Hong Kong University of Science and Technology, Hong Kong, China
}

\date{\today}

\begin{abstract}

Lines of exceptional points are robust in the 3-dimensional non-Hermitian parameter space without requiring any symmetry. However, when more elaborate exceptional structures are considered, the role of symmetry becomes critical. One such case is the exceptional chain (EC), which is formed by the intersection or osculation of multiple exceptional lines (ELs). In this study, we investigate a non-Hermitian classical mechanical system and reveal that a symmetry intrinsic to second-order dynamical equations, in combination with the source-free principle of ELs, guarantees the emergence of ECs. This symmetry can be understood as a non-Hermitian generalized latent symmetry, which is absent in prevailing formalisms rooted in first-order Schrödinger-like equations and has largely been overlooked so far. We experimentally confirm and characterize the ECs using an active mechanical oscillator system. Moreover, by measuring eigenvalue braiding around the ELs meeting at a chain point, we demonstrate the source-free principle of directed ELs that underlies the mechanism for EC formation. Our work not only enriches the diversity of non-Hermitian exceptional point configurations, but also highlights the new potential for non-Hermitian physics in second-order dynamical systems.

\end{abstract}

\maketitle

\textit{Introduction.---}Pioneered by the study of topological insulators in electronic systems~\cite{topology-review_Kane_2010_Rev.Mod.Phys.,topology-review_Zhang_2011_Rev.Mod.Phys.}, topology has brought a conceptual revolution sweeping across diverse fields of physics, including classical wave systems such as photonics~\cite{topologicalphotonics-review_Soljacic_2014_NaturePhoton,topologicalphotonics-review_Carusotto_2019_Rev.Mod.Phys.}, acoustics~\cite{ma2019Topological,xue2022Topological}, and elastic waves~\cite{kane2014Topological,susstrunk2016Classification,yoshida2019Exceptional,wang2015Topological,susstrunk2015Observation,fruchart2020Dualities}. On a different frontier, the development of non-Hermitian physics has recently been merging with topological phases ~\cite{shen2018Topological,kawabata2019Symmetry,kawabata2019Classification,wojcik2020Homotopy,wang2021Topological,hu2021Knots,yang2021fermion}. Due to the experimental advantages, classical systems are becoming powerful testbeds for non-Hermitian topological phenomena~\cite{ding2016Emergence,brandenbourger2019Nonreciprocal,tang2020Exceptional,ghatak2020Observation,tang2021Direct,zhang2021Acoustic,gao2021NonHermitian,zhang2021Observationa,hu2021NonHermitian,tang2022Experimental,baconnier2022Selective,liu2022Experimental,zhang2023Observation}. 
Exceptional points (EPs) are inherently non-Hermitian band singularities~\cite{heiss2004Exceptional}, where both the eigenvalues and eigenvectors of different bands coalesce. Many interesting topological properties arise from EPs and have been shown to be the foundations of diverse promising applications~\cite{miri2019Exceptional,ozdemir2019Parity,shankar2022Topological,bergholtz2021Exceptional,ding2022Review}. 
From a co-dimension consideration, the order-2 EPs can form stable exceptional lines (ELs) without the need for any symmetry in a three-dimensional (3D) parameter space, and the ELs can be further linked or knotted non-trivially~\cite{carlstrom2018Exceptional,yang2019NonHermitiana, zhang2020Bulkboundary,he2020Double,carlstrom2019Knotted,lee2020Imaging,zhang2021Tidal,wang2021Simulating}.
However, as a typical EL configuration, the
exceptional chain (EC)~\cite{cerjan2018Effects,yan2021Unconventional,zhang2022symmetry} formed by several connecting or osculating ELs is different from other EL morphologies because the stable existence of ECs demands symmetry protection. 
A recent study~\cite{zhang2022symmetry} has suggested that an EL can be assigned with an orientation determined by the complex-eigenvalue braiding around an EL, which 
is much similar to the right-hand rule describing the relation between DC electrical current and the magnetic fields it produces.
The oriented ELs are source-free in the parameter space.
It follows that stable ECs can exist as a consequence of the source-free principle of ELs in conjunction with certain symmetries, such as mirror, mirror-adjoint, and $C_2T$ symmetries~\cite{zhang2022symmetry,yan2021Unconventional}.
This interesting mechanism for stabilizing ECs has not yet been experimentally demonstrated or verified.
In this letter, we study the stable existence of ECs using a non-Hermitian mechanical model with three synthetic dimensions. Unlike most non-Hermitian models that are based on Schr\"odinger-like first-order differential equations, the mechanical oscillators are described by second-order differential equations (SDEs)~\cite{tisseur2001Quadratic}. 
We reveal that SDEs can exhibit special non-Hermitian symmetries that become hidden after linearization  to a Schr\"odinger-like Hamiltonian form, a process commonly employed in the study of topological mechanics~\cite{susstrunk2016Classification,yoshida2019Exceptional}. 
Via generalizing the recently proposed notion of latent symmetries~\cite{rontgen2021Latent,morfonios2021Flat,rontgen2023Hidden} to non-Hermitian scenarios, we demonstrate that these specific symmetries of the SDEs are essentially novel non-Hermitian latent symmetries of the linearized Hamiltonian, and they play a crucial role in protecting the stable ECs. 
 By constructing mechanical oscillators with active components~\cite{wang2022NonHermitian,wang2022Extended}, we experimentally realized an EC and characterized its topological features. Our results not only demonstrate the unique properties of ECs but also highlight the distinctive potential of SDE systems in the future study of non-Hermitian physics.




\textit{SDE for non-Hermitian oscillators.---}
For a mechanical system consisting of $N$ coupled harmonic oscillators, the dynamics is governed by the SDE~\cite{tisseur2001Quadratic}
\begin{equation}
     \mathbf{M}\frac{d^2}{dt^2} {X}(t) = - \mathbf{K}  {X}(t) -\mathbf \Gamma \frac{d}{dt} {X}(t),
\label{Eq-Newton}
\end{equation}
wherein $X(t) = (x_1(t), x_2(t),...,x_N(t))^\intercal$ represents the displacements of $N$ oscillators, and $\mathbf{M}$, $\mathbf{K}$, and $\mathbf{\Gamma}$ are real matrices representing the mass, stiffness, and damping matrices, respectively. $\mathbf{M}$, $\mathbf{\Gamma}$ are diagonal matrices, and $\mathbf{K}$ is symmetric (asymmetric) when the system is reciprocal (non-reciprocal). 
For time-harmonic solutions $x_i(t)=a_i e^{-i\omega t}$, the eigenfrequency $\omega$ obeys $Q(\omega)|\psi\rangle=(\omega^2\mathbf{M}  -  \mathbf{K} + i\omega \mathbf{\Gamma}) |\psi\rangle = 0 $,
where $|\psi\rangle= (a_1, a_2,..., a_N)^\intercal$ denotes the amplitudes of the oscillators and $Q(\omega)$ with real coefficient matrices is called a real quadratic matrix polynomial (QMP).
The quadratic eigenvalue problem (QEP) for $Q(\omega)$ is to find right $|\psi_n^r\rangle$  and left $|\psi_n^l\rangle$ eigenvectors and an eigenfrequency $\omega_n$ satisfying
\begin{equation} \label{Eq-QEPdefine}
    Q(\omega_n)|\psi_n^r\rangle =0, \quad \langle \psi_n^l | Q(\omega_n) = 0,
\end{equation}
Based on the customary approach in topological mechanics~\cite{susstrunk2016Classification,yoshida2019Exceptional}, the $N$-by-$N$ QEP can be transformed to a $2N$-dimensional linear eigenvalue problem by a standard mathematical technique of reducing the order of SDEs ~\cite{suppl2}: $\mathcal{H}|\Psi_n\rangle = i\left( {\begin{array}{*{20}{c}}
\mathbf{0}&\mathbf{1}\\
{ - \mathbf{M}^{-1}{\mathbf{K}}}&{ - \mathbf{M}^{-1}{\mathbf{\Gamma }}}
 \end{array}} \right) |\Psi_n\rangle = \omega_n |\Psi_n\rangle $ with $|{\Psi_n} \rangle=\left( |\psi^r_n\rangle, \frac{d}{dt}|\psi^r_n\rangle \right)^\intercal=\left( |\psi^r_n\rangle, -i\omega_n|\psi^r_n\rangle \right)^\intercal$,  
which takes the typical Hamiltonian formalism prevailing in tight-binding theories systems. 

The real QMP of any mechanical system possesses an intrinsic symmetry $Q(\omega)^*=Q(-\omega^*)$, which indicates for any right eigenvector $|\psi^r_n\rangle$ satisfying Eq.~\eqref{Eq-QEPdefine}, $|\psi_n^r\rangle^*$ is also a right eigenvector corresponding to the eigenvalue $-\omega_n^*$: $ Q(-\omega_n^*) |\psi_n^r\rangle^*=0$. Therefore, for a $N$-by-$N$ real QEP with a real line gap at $\mathrm{Re}(\omega)=0$, the $2N$ eigenvalues must come in pairs $(\omega_n, -\omega_n^*)$, implying that only $N$ eigensolutions of the system are independent. 
Indeed, by linearization, this intrinsic symmetry of real QMP manifests as a particle-hole symmetry of the corresponding $2N$-dimensional Hamiltonian $\mathcal H$: $\mathcal{H}^* = -\mathcal{H}$, protecting the pairwise eigenvalues~\cite{kane2014Topological,susstrunk2016Classification,yoshida2019Exceptional}. 
In the following, we focus on the upper $N$ bands with $\mathrm{Re}(\omega_n)>0$ because only the positive frequency (PF) modes are physically observable.


\begin{figure}[b!]
\includegraphics[width=0.49\textwidth]{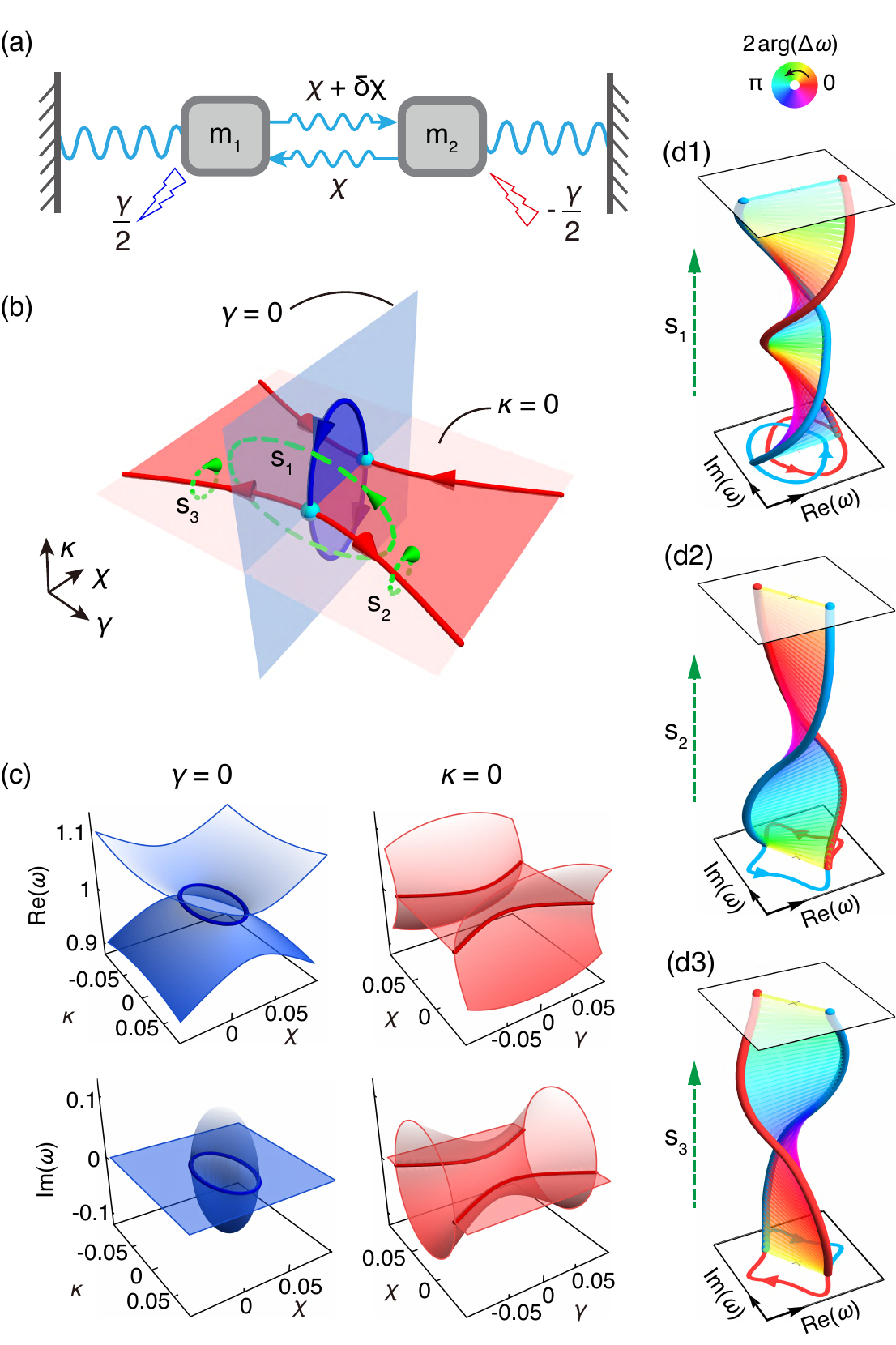}
\caption{\label{Fig-theory} 
(a) Schematic of the  nonreciprocally coupled oscillator model.
(b) An orthogonal EC formed by directed ELs in the 3D parameter space. (c) Eigenfrequencies of the two bands with positive real parts on planes $\gamma=0$ and $\kappa=0$. (d1-d3) Complex-eigenvalue braiding along the loops $S_{1,2,3}$ (green dashed lines) in (b), where the color of the bars connecting the two eigenfrequencies' trajectories represents twice the phase of the relative eigenfrequency $\Delta \omega =\omega_1-\omega_2$.
The fixed parameters used for plotting: $m_0=\bar\kappa=1$, $\delta\chi=-0.05$.}
\end{figure}

\textit{Theory of mechanical symmetry-protected ECs.---}
In a model with $N=2$ [Fig.~\ref{Fig-theory}(a)], a nonreciprocal spring connects two nonconservative oscillators ($m_{1,2}=m_0$) subject to intrinsic damping and gain, respectively, 
and their motions are described by a 2-by-2 QMP, where the coefficient matrices can be expressed as $\mathbf{M}= \mathrm{diag}(m_0,m_0)$,
\begin{equation}\label{Eq-TheorModel}
\mathbf K = \left( {\begin{array}{*{20}{c}}
{{\bar{\kappa}+\kappa/2}}&{{-\chi}}\\
{{-\chi-\delta\chi}}&{{\bar{\kappa}-\kappa/2}}
\end{array}} \right), 
\  
\mathbf \Gamma =\left( {\begin{array}{*{20}{c}}
{{ \gamma/2}}&{0}\\
{0}&{{-\gamma/2}}
\end{array}} \right). 
\end{equation}
The diagonal terms of $\mathbf{K}$ determine the natural frequencies of the two isolated oscillators, and $\kappa$ represents their difference. 
The off-diagonal terms are different, which indicates nonreciprocal, direction-dependent coupling. 
The strictly opposite diagonal terms of $\mathbf{\Gamma}$ indicate balanced damping ($\gamma/2$) and gain ($-\gamma/2$) in the two oscillators.

We treat $\gamma$ (relative damping), $\chi$ (coupling strength) and $\kappa$ (onsite stiffness difference) as three synthetic dimensions spanning a 3D parameter space $(\gamma,\chi,\kappa)$ and keep $\bar\kappa$ and $\delta\chi$ fixed, the model possesses two symmetries in the parameter space~\cite{note2}. 
The first is ``$\gamma$-symmetry", resulting from the balanced gain and loss of the two oscillators,
\begin{equation}\label{Eq-Symetery-gamma}
    Q^*(\omega,\gamma)=Q(\omega^*,-\gamma).
\end{equation}
Applying complex conjugate $\mathcal{K}$ on the Eq.~\eqref{Eq-QEPdefine}, we obtain $ Q^*(\omega_n,\gamma)|\psi_n^r\rangle^* =  Q(\omega_n^*,-\gamma)|\psi_n^r\rangle^* = 0$,
indicating that for any eigenfrequency $\omega_n$ at $\gamma$, there is always a corresponding eigenfrequency $\omega_n^*$ at $-\gamma$.
It follows that the eigenfrequencies on the high-symmetry plane $\gamma=0$ either take real values (outside the ring) or form complex conjugate pairs (inside the ring), as shown in the left panel of Fig.~\ref{Fig-theory}(c). And the $\gamma=0$ plane [Fig.~\ref{Fig-theory}(b)] is divided into exact (light blue) and broken (dark blue) phases by an exceptional ring (ER). We should note that when $\gamma=0$, the non-Hermiticity of the system comes purely from the nonreciprocity $\delta\chi$, which determines the size of the ER. 

In addition, since $\sigma_x \mathbf{K^\dagger(\kappa)} \sigma_x = \mathbf{K(-\kappa)}$, $\sigma_x \mathbf{\Gamma}^\dagger \sigma_x = -\mathbf{\Gamma}$  ($\sigma_x$ is the first Pauli matrix),
the QMP has a second ``$\kappa$-symmetry",
\begin{equation}\label{Eq-Symetery-kappa}
    \sigma_x Q^\dagger(\omega, \kappa) \sigma_x =  Q(\omega^*, -\kappa).
\end{equation}
Applying transpose conjugate on the QEP, we have
$\langle \psi_n^r|Q^\dagger(\omega_n, \kappa) =
    \langle \psi_n^r|\sigma_x Q(\omega_n^*, -\kappa) = 0$,
signifying that $\langle \psi_n^r|\sigma_x$ is a left eigenvector of $Q(\omega_n^*, -\kappa)$. Thus, for any eigenstate with frequency $\omega_n$ at $\kappa$, there is always an eigenstate with $\omega_n^*$ at $-\kappa$. 
On the high-symmetry plane $\kappa=0$, the eigenvalues are either real or form complex conjugate pairs [right panel of Fig.~\ref{Fig-theory}(c)], and the ELs are fixed on the $\kappa=0$ plane, serving as the phase transition boundary between the exact (light red) and broken (dark red) phases [Fig.~\ref{Fig-theory}(b)].
Finally, under the protection of the symmetries \eqref{Eq-Symetery-gamma} and \eqref{Eq-Symetery-kappa}, 
the ELs on two orthogonal planes connect together [cyan dots in Fig.~1(b)], 
forming a symmetry-protected EC in the 3D parameter space.

\begin{figure}[b!]
\includegraphics[width=0.49\textwidth]{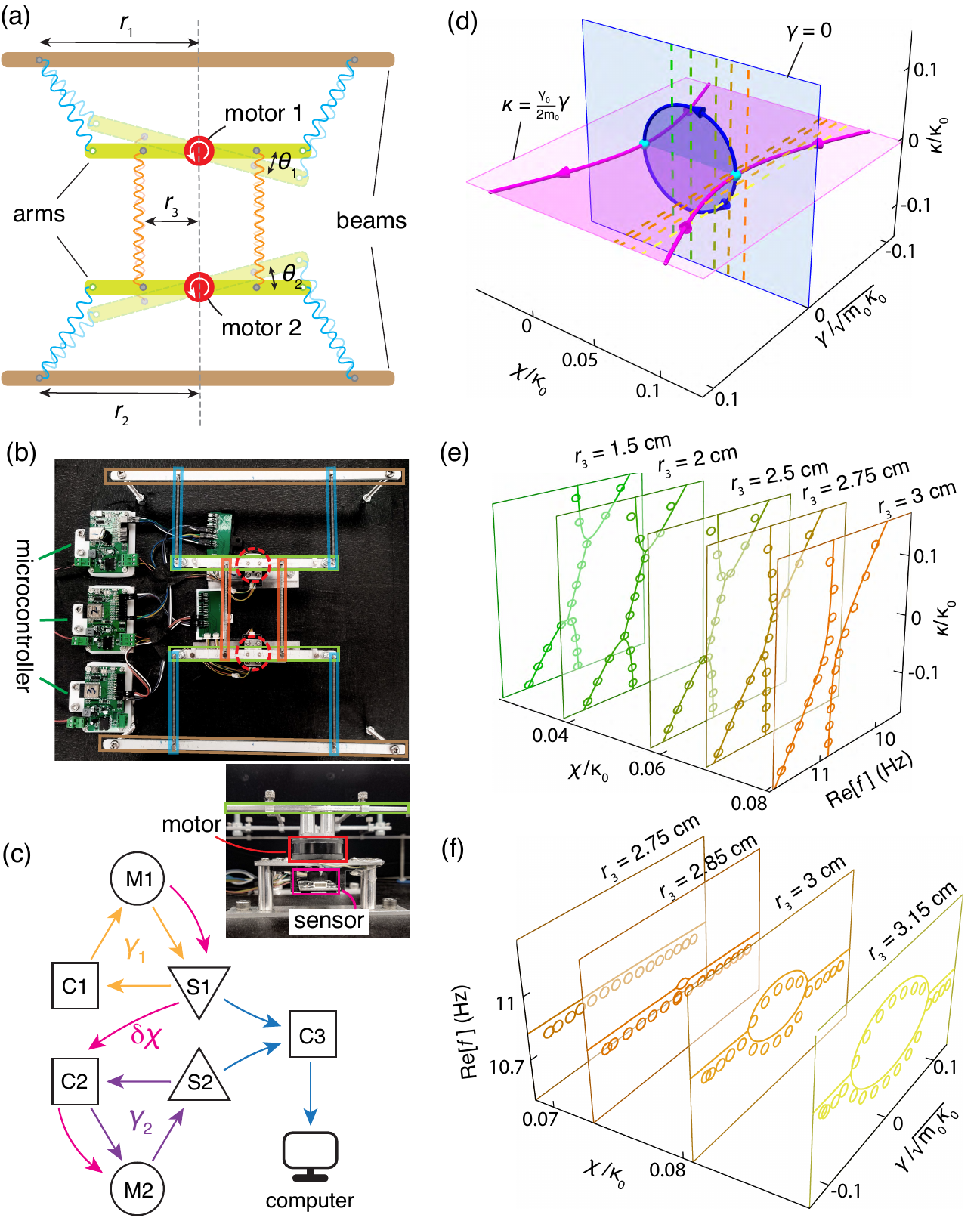}
\caption{\label{Fig-experiment} (a) Schematic of the active oscillators. (b) Experimental setup. (c) Schematic of the control system where the microcontrollers (C1, C2) process the signals from the sensors (S1, S2) and send commands to the motors (M1, M2), and the instantaneous rotation angles of the motors are recorded by computer via the microcontroller C3. (d) EC realized by the experimental model, where an ER (blue tube) fixed on the plane $\gamma=0$ connects two out-of-plane ELs (magenta tubes). The measured eigenfrequencies for different connection positions $r_3$ on the planes of (e) $\gamma=0$ and (f) $\kappa = \frac{\gamma_0}{2m_0}\gamma$. 
The fixed parameters retrieved using Green's function method: $\delta\chi=-0.073 \kappa_0$, $\gamma_0/\sqrt{m_0 \kappa_0}=0.085$.}
\end{figure}


The $\gamma$-symmetry \eqref{Eq-Symetery-gamma} of $Q(\omega)$ can be mapped to an antiunitary symmetry of the  4-by-4 effective Hamiltonian $\mathcal{H}$: $U_\gamma \mathcal{H}(\gamma) U_\gamma^{-1} 
= \mathcal{H}(-\gamma)$ with $U_\gamma= \tau_z \otimes \sigma_0 \mathcal{K} $. However, the $\kappa$-symmetry \eqref{Eq-Symetery-kappa} cannot be directly transformed to a usual symmetry of $\mathcal{H}$. Rather, we revealed that it corresponds to a non-Hermitian latent symmetry of $\mathcal{H}$, i.e., $\sigma_x(\mathcal{H}(-\kappa)^n)_{BR}\sigma_x=(\mathcal{H}(\kappa)^n)_{BR}^\dagger$ for any $n\in\mathbb{N}$ with $(\mathcal{H}^n)_{BR}$ denoting the bottom right block of $\mathcal{H}^n$~\cite{suppl2}, which generalizes the notion of latent symmetry recently proposed in Hermitian systems~\cite{rontgen2021Latent,morfonios2021Flat,rontgen2023Hidden}. This observation suggests that some symmetries of the original SDE become difficult to recognize in the effective Hamiltonian after linearization, and the origin of coalescence of eigenvalues becomes less obvious. Hence, directly analyzing the symmetries of the SDE is more natural in some scenarios for mechanical systems, though the effective Hamiltonian might be preferred for computational purposes. Moreover, the non-Hermitian latent symmetry also plays a significant role in characterizing the generalized crystalline symmetries appearing in non-Hermitian mechanical lattices (see Supplemental Materials~\cite{suppl2} for details).





In non-Hermitian systems, when a pair of bands form EPs, their eigenvalues, $\omega_m$ and $\omega_n$, will generally braid about each other along a loop $S$ encircling the EPs~\cite{wojcik2020Homotopy,wang2021Topological}. 
To characterize this eigenvalue braiding, a half-quantized topological invariant called energy vorticity~\cite{shen2018Topological} has been introduced,
$\nu_{mn}(S)=\frac{1}{2 \pi} \oint_{S} d \mathbf{g} \cdot \nabla_{\mathbf{g}} \arg \left[\omega_{m}(\mathbf{g})-\omega_{n}(\mathbf{g})\right]$.
As shown in Figs.~\ref{Fig-theory}(d1)-(d3), the eigenfrequency braiding of the two PF bands along the loop [dashed green lines in Fig.~\ref{Fig-theory}(b)] 
exhibits that twice the energy vorticity equals the net number of times the two bands braid. 
The energy vorticities carried by the loops are $\nu_{12}(S_1)=1$, $\nu_{12}(S_2)=1/2$ and $\nu_{12}(S_3)=-1/2$ respectively, where the sign of $\nu_{12}$ denotes the handedness of the braid and endows the encircled ELs with a positive orientation (which are indicated by arrows on ELs) in compliance with the right-hand rule~\cite{zhang2022symmetry}. For example, the positive (negative) sign of $\nu_{12}(S_2)$ ($\nu_{12}(S_3)$) indicates the EL has the same (reverse) direction as the rightward normal unit vector of the loop $S_2$ ($S_3$). And based on the direction of the ELs, we can prove a generalized source-free principle of the PF ELs in mechanical systems with the intrinsic particle-hole symmetry~\cite{zhang2022symmetry,suppl2}: the number of PF ELs flowing into a junction must equal the number of PF ELs flowing out. This behavior is clearly observed for the oriented ELs near the chain points in Fig.~\ref{Fig-theory}(b).
It is noteworthy that the source-free principle is crucial for the stable existence of ECs. Specifically, it elucidates why the ELs, constrained in distinct planes by the two symmetries, must osculate each other, and further implies that the chain point is robust against the breaking of either symmetry~\cite{suppl2}.

\textit{Experimental realization of ECs.---}
Our experiment is performed using active mechanical oscillators. 
In the setup [Fig.~\ref{Fig-experiment}(a-c)], two rotational arms (green) connected to brushless DC motors (red) are attached to rigid beams (brown) at $r_{1,2}$ by two identical springs (blue). Another two identical springs (orange) separated by $2r_3$ connect the two rotational arms in parallel.
At small-angle approximation, $\sin(\theta_n) \approx \theta_n$, the vibration equations of the active oscillators is described by the SDE in Eq.~\eqref{Eq-Newton}, where $X(t)=(\theta_1(t),\theta_2(t))^\intercal$ represents the oscillation angles of the two rotational arms with equal moments of inertia, $\mathbf{M}= \mathrm{diag}(m_0,m_0)$, and
the stiffness and damping matrices take the forms~\cite{suppl2},
\begin{equation}\label{Eq-ExperModel}
\begin{aligned}
\mathbf K &= \left( {\begin{array}{*{20}{c}}
{{\kappa_{0}+\chi}}&{{-\chi}}\\
{{-\chi-\delta\chi}}&{{\kappa_{0}-\kappa+\chi}}
\end{array}} \right), 
\\ 
\mathbf \Gamma &=\left( {\begin{array}{*{20}{c}}
{{ \gamma_0+\gamma/2}}&{0}\\
{0}&{{\gamma_0-\gamma/2}}
\end{array}} \right).
\end{aligned}
\end{equation}
$\mathbf{\Gamma}$ is loss-biased by adding a constant loss $\gamma_0$ to avoid instability issues encountered when the system is in the gain regime. 
The parameters $\chi$, $\kappa_0$, and $(\kappa_0-\kappa)$ are determined by conservative torques applied by the orange springs, the upper blue springs, and the lower blue springs in Fig.~\ref{Fig-experiment}(a), respectively. The parameters $\delta \chi$ and $\gamma$ are implemented as non-conservative torques exerted by the motors. To this end, the instantaneous oscillation angle is measured using a Hall sensor in real time. The angle is used as feedback to drive the microcontrollers that are programmed to command the motors to exert the desired torques [see Fig.~\ref{Fig-experiment}(b)]. 

\begin{figure*}[t!]
\includegraphics[width=0.95\textwidth]{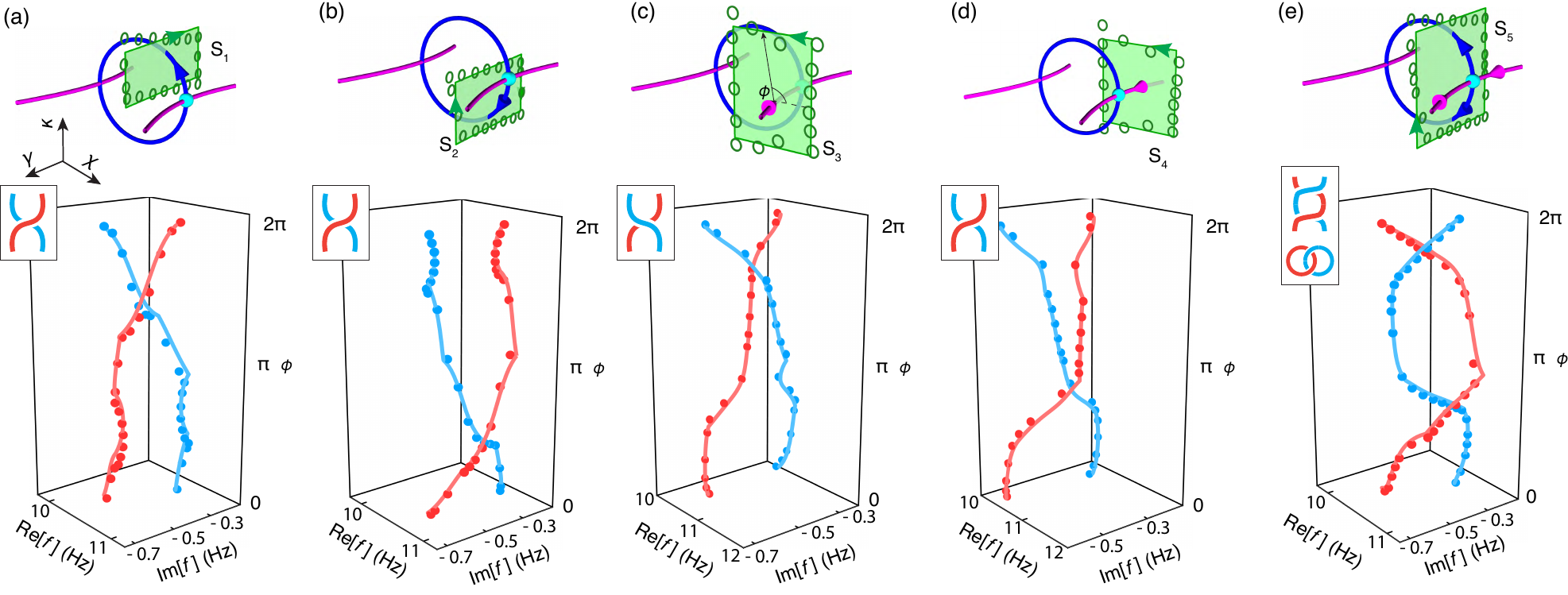}
\caption{\label{Fig-windingloop} 
Complex eigenfrequency braiding (red and cyan lines/dots in the lower panels) along rectangular loops (green lines/dots in the upper panels) encircling different ELs, where $\phi\in[0,2\pi]$ represents the sweeping parameter of the loops. The braiding diagrams are plotted in the insets. The details about the parameters retrieval can be found in Ref.~\cite{suppl2}. 
}
\end{figure*}

During the experiments, $\kappa_0$ is kept constant by fixing $r_1$, and $\kappa$ and $\chi$ are tuned by changing $r_2$ and $r_3$, respectively. $\gamma$ and $\delta \chi$ are tuned by velocity-dependent and angle-dependent torques applied to the motors [see details in~\cite{suppl2}], as shown in Fig.~\ref{Fig-experiment}(c). 
Compared with the theoretical model in Fig.~\ref{Fig-theory}, 
the biased loss $\gamma_0$  reduces the whole-space $\gamma$-symmetry shown by Eq.~\eqref{Eq-Symetery-gamma}  to a subspace symmetry on the $\gamma=0$ plane in
the experimental system:
\begin{equation}
    Q^*(\omega)=Q(\omega^*-i \gamma_0/m_0),
\end{equation}
Meanwhile, the $\kappa$-symmetry in Eq.~\eqref{Eq-Symetery-kappa} is reduced to a subspace symmetry on the oblique plane $\kappa=\frac{\gamma_0}{2m_0}\gamma$~\cite{suppl2}:
\begin{equation}
    \sigma_x{Q}^\dagger({\omega})\sigma_x={Q}({\omega}^*-i \gamma_0/m_0),
\end{equation}
which can be regarded as a non-Hermitian latent symmetry for the linearized Hamiltonian of the QMP on that plane. Remarkably, the two subspace symmetries can also guarantee that the eigenfrequencies on the two planes either appear in pairs $(\omega_n,\omega_n^* - i\gamma_0/m_0)$ or have a common imaginary part $-i\gamma_0/2m_0$, corresponding to the broken and exact phases, respectively.
Accordingly, despite the presence of background loss, ELs are still rigorously fixed at the boundaries of the two phases on the two planes and are joined on their intersection line $\gamma=\kappa=0$, hence forming an EC.

In Fig.~\ref{Fig-experiment}(d), we plot the EC configuration of the experimental model with a fixed nonreciprocal strength $\delta\chi = -0.073 \kappa_0$, which is retrieved using the Green's function method~\cite{suppl2}. The inclination of the $\kappa=\frac{\gamma_0}{2m_0}\gamma$ plane (magenta) is about $2.5\degree$ in the Figure due to the small value of $\gamma_0$ in the experiments.
Figures~\ref{Fig-experiment}(e) and (f) display the measured eigenfrequencies of the PF bands for varying $\chi$ [dashed lines in  Fig.~\ref{Fig-experiment}(d)] on the $\gamma=0$ and $\kappa=\frac{\gamma_0}{2m_0}\gamma$ planes. 
The EPs can be clearly identified from the measured band structures, and the EC formed by the ELs on the two planes becomes obvious by tracing the positions of the EPs.

By measuring eigenfrequency braidings, we experimentally verified the orientations of the ELs assigned in Fig.~\ref{Fig-experiment}(d) and validated the formation of the EC. As shown in Fig.~\ref{Fig-windingloop}(a-b), we varied the parameters $\gamma$ and $\kappa$ along two rectangular clockwise loops encircling the upper and lower semi-ER on the plane $\chi = 0.057 \kappa_0$. The measured eigenfrequency braidings along the loops are shown in the lower panels, from which the energy vorticities are obtained as $\nu_{12}(S_1) =\nu_{12}(S_2)= 1/2$. 
According to the right-hand rule, the upper and lower semi-ERs both flow outward from the chain point (cyan dot). 
As shown in Fig.~\ref{Fig-windingloop}(c-d), we chose another two loops $S_3$ and $S_4$ on the planes of $\gamma/\sqrt{m_0 \kappa_0} = 0.043$ and of $\gamma/\sqrt{m_0 \kappa_0} = -0.03$. The non-trivial eigenfrequency braiding implies the presence of an EL enclosed in each loop, and the orientations of the ELs are both toward the chain point as per the signs of energy vorticities $\nu_{12}(S_3) = -1/2$ and  $\nu_{12}(S_4) = 1/2$, which confirms the source-free principle at the chain points. Lastly, we consider a loop encircling the chain point (cyan dot) shown in Fig.~\ref{Fig-windingloop}(e), and find the eigenfrequencies braid twice, which form a Hopf link after connecting the eigenfrequency trajectories head-to-tail.

\textit{Discussions.---}
We have also studied the physical effect induced by the ECs in 3D mechanical lattices and have discovered that the unique concurrent linear intersections of the real and imaginary parts of bands near the EC point can give rise to the exotic phenomenon of splitting and reshaping a spatially localized pulse into two oppositely propagating needle pulses~\cite{suppl2}. This discovery demonstrates that ECs also possess functional properties with promising application potential. 

In summary, we revealed that the non-Hermitian latent symmetries intrinsic to SDE governing  mechanical systems play a pivotal role in the formation of ECs, and such a symmetry-protected EC is experimentally confirmed using coupled active mechanical oscillators.
By measuring the eigenfrequency braiding around ELs, we identified the orientations of the ELs and verified the generalized source-free principle for the PF ELs in SDE problems. Our study not only demonstrates the stable existence of EC as a different kind of topological configuration, but also opens the door toward the topological effects inherent to the SDEs, which govern a broad class of systems encompassing classical mechanics, classical waves, electricity ~\cite{helbig2020Generalized,liu2021NonHermitian,hu2022nonhermitian}, optomechanical~\cite{li2021NonHermitian}, and micro-electromechanical~\cite{xu2022Nonreciprocal} systems.


\begin{acknowledgements}
\textit{Acknowledgements.}---We thank Prof.~Zhao-Qing Zhang, Drs. Hongwei Jia, Jing Hu and Yixin Xiao for the helpful discussions.
This work is supported by the Research Grants Council of Hong Kong (AoE/P-502/20, R6015-18, RFS2223-2S01, 16307420, 12302420, 12301822), the Croucher Foundation (CAS20SC01), and the National Key R\&D Program of China (2022YFA1404400).
\end{acknowledgements}

\bibliography{reference.bib}

\end{document}


\title{
Supplemental Materials for ``Experimental realization of stable exceptional chains protected by non-Hermitian latent symmetries unique to mechanical systems''}

\author{Xiaohan Cui}
\thanks{These authors contributed equally to this work.}
\affiliation{%
Department of Physics, The Hong Kong University of Science and Technology, Hong Kong, China
}%

\author{Ruo-Yang Zhang}
\thanks{These authors contributed equally to this work.}
\affiliation{%
Department of Physics, The Hong Kong University of Science and Technology, Hong Kong, China
}%

\author{Xulong Wang}
\affiliation{%
Department of Physics, Hong Kong Baptist University, Kowloon Tong, Hong Kong, China
}%
\author{Wei Wang}
\affiliation{%
Department of Physics, Hong Kong Baptist University, Kowloon Tong, Hong Kong, China
}
\author{Guancong Ma}%
\email{phgcma@hkbu.edu.hk}
\affiliation{%
Department of Physics, Hong Kong Baptist University, Kowloon Tong, Hong Kong, China
}%
\author{C. T. Chan}%
\email{phchan@ust.hk}
\affiliation{%
Department of Physics, The Hong Kong University of Science and Technology, Hong Kong, China
}

\maketitle



\setcounter{secnumdepth}{3}
\setcounter{equation}{0}
\setcounter{figure}{0}
\renewcommand{\theequation}{S\arabic{equation}}
\renewcommand{\thefigure}{S\arabic{figure}}
\newcommand\Scite[1]{[S\citealp{#1}]}

\renewcommand{\thetheorem}{\arabic{theorem}.}
\renewcommand{\thedefinition}{\arabic{definition}.}


\vspace{-20pt}

\tableofcontents

\section{Generalized source-free principle for positive-frequency exceptional lines in mechanical systems}

\begin{figure*}[b!]
\includegraphics[width=0.8\textwidth]{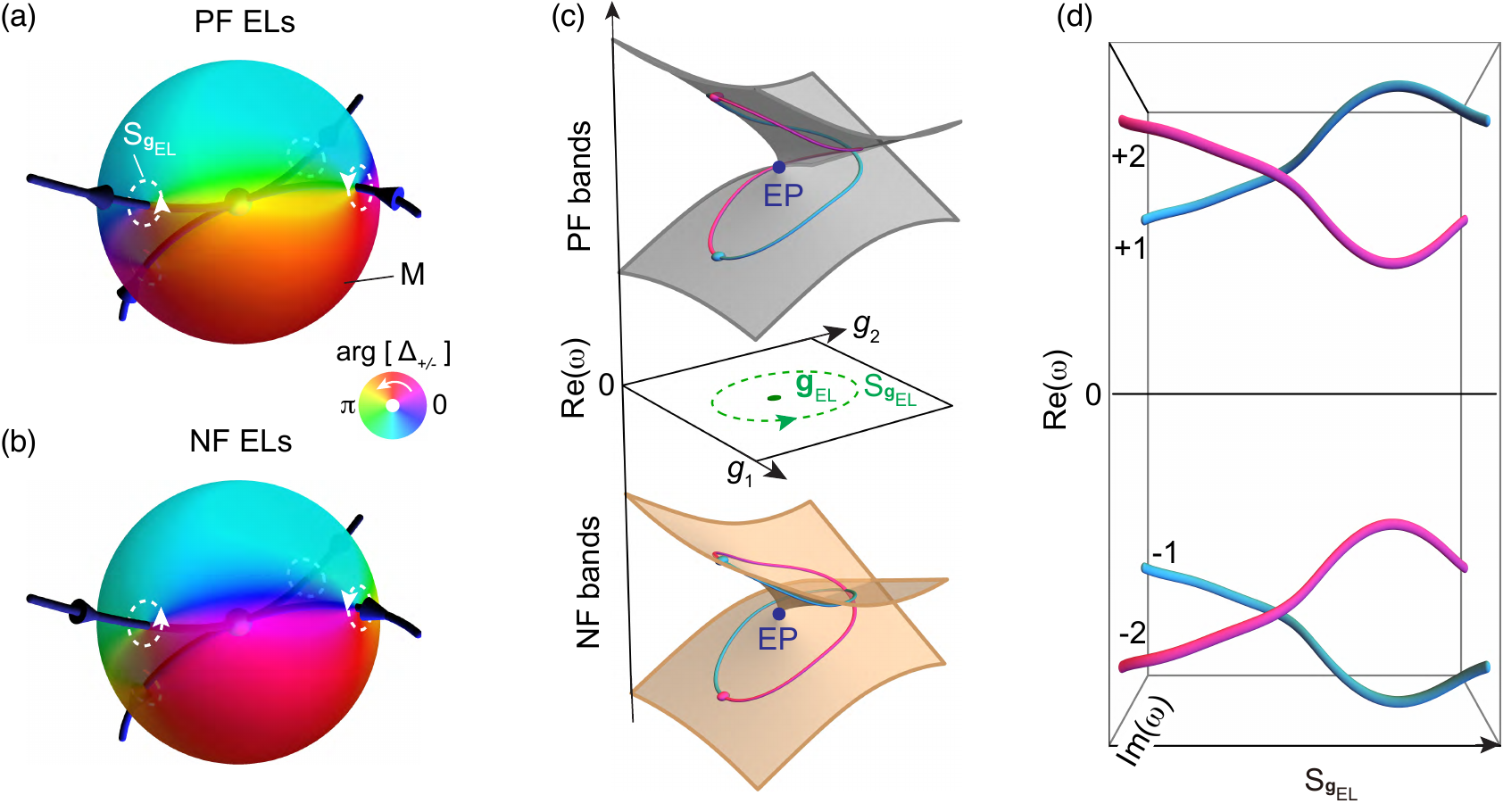}
\caption{\label{Fig-source-free-principle}
The generalized source-free principle for (a) PF ELs and (b) NF ELs, where the colors on the sphere enclosing the chain point denote the arguments of the PF and NF discriminants, respectively. (c) Spectral symmetry $\omega_n(\mathbf{g})=-\omega_{-n}(\mathbf{g})^*$ of PF and NF bands induced by the intrinsic particle-hole symmetry, where the PF and NF EPs always appear in pairs at the same point in the parameter space. (d) Spectral symmetry guarantees that PF and NF bands along a loop $S_{\mathbf{g}_\mathrm{EL}}$ encircling 
the pair of PF and NF EPs always braid in opposite manners, giving rise to opposite energy vorticies.
}
\end{figure*}
In general non-Hermitian systems, the topology of exceptional lines (ELs) can be captured by the discriminant $\Delta_p$ of the characteristic polynomial $p(\omega)=\det[\omega-\mathcal{H}]$ of the Hamiltonian $\mathcal{H}$. Since the discriminant $\Delta_p$ is a complex number determined by the eigenvalues $\omega_n$: $\Delta_p=\prod_{m<n}(\omega_m-\omega_n)^2$, the phase winding number of $\Delta_p$ along a loop $S$ in the parameter space introduces an integer topological invariant, called discriminant number (DN)~\cite{yang2021fermion}, 
\begin{equation}\label{DN}
    \mathcal{D}(S)=\frac{1}{2\pi}\oint_S d\mathbf{g}\cdot\nabla_{\mathbf{g}}\arg\Delta_p(\mathbf{g})=2\sum_{m<n}\nu_{mn}(S),
\end{equation}
which precisely equals twice the sum of the energy vorticities of all pairs of bands, where $\mathbf{g}$ denotes the position vector in the 3D parameter space. In particular, for a loop $S_\mathrm{EL}$ enclosing a single EL formed by bands $n$ and $n+1$,  $\mathcal{D}(S_\mathrm{EL})=2\nu_{n,n+1}(S_\mathrm{EL})$, hence the orientation of the EL can alternatively be determined by  the DN $\mathcal{D}(S_\mathrm{EL})$~\cite{zhang2022symmetry}. Moreover, it has been proved that the directed ELs obey the \textbf{source-free principle} in 3D parameter space~\cite{zhang2022symmetry}, meaning that for an arbitrary oriented and closed surface in the space, the numbers of ELs flowing into and flowing out from the surface are always balanced. This is also the key principle ensuring the formation of exceptional chains.

However, in mechanical systems, the intrinsic particle-hole symmetry restricts the discriminant of the characteristic polynomial $p(\omega)=\det[\omega-\mathcal{H}]=\frac{(-1)^N}{\det\qty[\mathbf{M}]}\det[Q(\omega)]$ to be real-valued $\Delta_p\in\mathbb{R}$. As a result, the DN along any loop is always trivial $\mathcal{D}(\Gamma)=0$. Physically, the spectral symmetry $\omega(\mathbf{g})_{n}=-\omega(\mathbf{g})_{-n}^*$ between the positive-frequency (PF) and negative-frequency (NF) bands induced by the particle-hole symmetry guarantees that  PF and NF ELs always appear pairwise in the parameter space, as illustrated in Fig.~\ref{Fig-source-free-principle}(a,b). In each pair, they spatially coincide exactly but have opposite orientations due to the opposite energy vorticities along a loop $S$ encircling the pair of ELs [see Fig.~\ref{Fig-source-free-principle}(c,d)]:
\begin{equation}
    \nu_{n,n+1}(S)=\frac{1}{2 \pi} \oint_{S} d \arg \left[\omega_{n}(\mathbf{g})-\omega_{n+1}(\mathbf{g})\right]=\frac{1}{2 \pi} \oint_{S} d \arg \left[-\omega_{-n}^*(\mathbf{g})-\omega_{-n-1}^*(\mathbf{g})\right]=-\nu_{-n,-n-1}(S),
\end{equation}
where $n(>0),n+1$ and $-n,(-n-1)$ label the two pairs of bands forming the PF EL and NF EL, respectively. Consequently, the source-free principle of ELs is trivially satisfied in mechanical systems because the contributions of a PF NL and its ``image'' NF EL are always canceled. 

Nevertheless, as long as the PF and NF bands are well separated by a real line gap at zero frequency $\omega=0$, we can construct a ``\textbf{positive-frequency discriminant}'' purely using the eigenfrequencies $\omega_n$ ($n>0$) with positive real parts:
\begin{equation}
    \Delta_+(\mathbf{g})=\prod_{0<m<n}\qty(\omega_m(\mathbf{g})-\omega_n(\mathbf{g}))^2,
\end{equation}
whose zeros correspond to the degeneracies of the PF bands. Given that $\mathcal{H}(\mathbf{g})$ is a continuous function of the parameters $\mathbf{g}$, all the eigenvalues of $\mathcal{H}$ are also continuous with respect to $\mathbf{g}$ (see Theorem 5.1 in Ref.~\cite{kato1995perturbation}), which further guarantees the PF discriminant $\Delta_+(\mathbf{g})$ to be a continuous single-valued function in the parameter space. Akin to the DN~\eqref{DN}, the phase winding number of $\Delta_{+}(\mathbf{g})$ along a closed loop $S$ is also an integer topological invariant, which we term \textbf{positive-frequency discriminant number} (PFDN),
\begin{equation}
    \mathcal{D}_+(S)=\frac{1}{2\pi}\oint_S d\mathbf{g}\cdot\nabla_{\mathbf{g}}\arg\Delta_+(\mathbf{g})=2\sum_{0<m<n}\nu_{mn}(S),
\end{equation}
which counts the net number of directed PF ELs passing through the loop $S$. And if the loop encircles solely one PF EL, we may also use the PFDN $\mathcal{D}_+(S)=\pm1$ to assign the orientation of the PF EL with the positively oriented tangent vector of the EL:
$
    \mathbf{t}_\mathrm{EL}=\mathcal{D}_+(S)\mathbf{t}_S
$, 
where $\mathbf{t}_S$ represents the tangent vector of EL with its positive direction determined by the right-hand rule of the directed loop $S$. 
Then, using the Poincar\`e-Hopf theorem for complex line bundles, we can generalize the source-free principle for all ELs~\cite{zhang2022symmetry} to the PF ELs in mechanical systems:
\begin{theorem}[Source-free principle for PF ELs]
For a quadratic matrix polynomial, as long as the PF and NF bands are well-separated by a real line gap at $\mathrm{Re}(\omega)=0$ on an oriented and closed surface $M$ in the 3D parameter space, the ELs passing through the surface obey the index theorem:
\begin{equation}
    \sum_{\mathbf{g}_\mathrm{EL}\in M}\mathcal{D}_+(S_{\mathbf{g}_\mathrm{EL}})\equiv 0,
\end{equation}
where $\mathbf{g}_\mathrm{EL}$ denote the points on $M$ where the PF ELs pierce through; around each degenerate point $\mathbf{g}_\mathrm{EL}$ on $M$, $S_{\mathbf{g}_\mathrm{EL}}$ denotes a small loop on $M$ encircling it and the positive direction of the loop is consistent with the outward normal of the surface. The index theorem manifests that \textbf{the number of  PF ELs flowing into the surface always equals that of the PF ELs flowing out of the surface}.
\end{theorem}
\begin{proof}
Regarding the PF discriminant $\Delta_+(\mathbf{g})$ as a complex function on the surface, we may introduce the corresponding Berry connection and Berry curvature away from the ELs ($\mathbf{g}\not\in\qty{\mathbf{g}_\mathrm{EL}}$) on the surface:
\begin{align*}
    \mathcal{A}(\mathbf{g}) &=-i\frac{\Delta_+^*\nabla_\mathbf{g}\Delta_+}{|\Delta_+|^2}=\nabla_\mathbf{g}\arg\Delta_+(\mathbf{g}),\qquad
    \mathcal{F}(\mathbf{g}) =\nabla_\mathbf{g}\times\mathcal{A}(\mathbf{g})=0.
\end{align*}
Therefore, using Stokes' theorem on the oriented surface $M$ excluding the areas near the degenerate points $\qty{\mathbf{g}_\mathrm{EL}}$ (i.e., the regions enclosed by the small loops $S_{\mathbf{g}_\mathrm{EL}}$), we obtain
\begin{equation}
    \sum_{\mathbf{g}_\mathrm{EL}\in M}\mathcal{D}_+(S_{\mathbf{g}_\mathrm{EL}})=\sum_{\mathbf{g}_\mathrm{EL}\in M}\oint_{S_{\mathbf{g}_\mathrm{EL}}}d\mathbf{g}\cdot\mathcal{A}(\mathbf{g})=\iint_{M-\qty{S_{\mathbf{g}_\mathrm{EL}}}} du\wedge dv\ \hat{\mathbf{n}}\cdot\mathcal{F}(\mathbf{g})\equiv 0,
\end{equation}
where $M-\qty{S_{\mathbf{g}_\mathrm{EL}}}$ represents the region on the surface excluding the areas enclosed by the small loops $S_{\mathbf{g}_\mathrm{EL}}$ [see e.g. the sphere in Fig.~\ref{Fig-source-free-principle}(a) excluding the areas inside the white dashed loops], $\hat{\mathbf{n}}$ denotes the outward normal of the surface and $du\wedge dv$ denotes the surface element.
\end{proof}
Similarly, by introducing an NF discriminant $\Delta_-=\Pi_{m<n<0}\qty(\omega_m-\omega_n)^2$, the source-free principle for NF ELs can also be proved [see Fig.~\ref{Fig-source-free-principle}(b)].  
According to the generalized source-free principle, we know that when several PF ELs meet at a junction, the inflow and outflow PF ELs must be balanced, which lays the foundation for the formation of exceptional chains in mechanical systems.


\section{Topological nature of symmetry-protected exceptional chains}
In this section, we will explain the relationship between eigenvalue braiding and the exceptional chain. Although the presence of eigenvalue braiding is irrespective of whether the ELs touch or not, the converse is not true. Indeed, the braiding (or equivalently the orientation of ELs) plays a crucial role in the formation of the chain of ELs under the combined action with certain symmetries. And the source-free principle of ELs serves as the foundation for comprehending the stability of ECs. Furthermore, we proved that the EC is a new symmetry-protected topological (SPT) gapless phase and proposed a topological invariant characterizing the distinct topology of the EC.

\subsection{Relationship between source-free principle and stability of ECs}
Let us elucidate the role of eigenvalue braidings and the source-free principle in stabilizing the EC through a step-by-step process. We consider a system with two symmetries $R_1, R_2$, where $R_1, R_2$ can be either the $\gamma$- and $\kappa$- symmetries of the synthetic dimension model or the $C_2T$ and (generalized) mirror-adjoint symmetries of the 3D lattice model (see Section.\ref{sec-3Dlattice}), and their symmetry-invariant planes $\Pi_1, \Pi_2$ are orthogonal. As shown in Fig. \ref{Fig-ECformation}(a), the red EL is confined by $R_1$ in its symmetry-invariant plane $\Pi_1$. Due to the second symmetry $R_2$, the EL lies symmetrically about the plane $\Pi_2$. Notably, $R_2$ also imposes spectral symmetry onto the eigenvalues: $\omega(\mathbf{g}) = \omega^*(\hat{R_2} \mathbf{g})$, which indicates that the eigenvalue braidings, characterized by the energy vorticity $\nu_{12}$ of the two relevant bands along the two $R_2$-symmetric loops $\Gamma$ and $R_2\Gamma$, must take opposite signs: $\nu_{12}(\Gamma) = -\nu_{12}(\hat{R_2}\Gamma)$. Hence, the red EL must reverse its orientation when crossing the middle plane $\Pi_2$. Consequently, both red half-ELs flow inward to the midpoint $K_0$, and, as depicted in Fig. \ref{Fig-ECformation}(b), the source-free principle of PF ELs requires the existence of two additional ELs (i.e., the two blue half-ELs) on $\Pi_2$ outflowing from $K_0$, thereby forming the EC. From this analysis, we can observe that the eigenvalue braiding acts as the quantized topological flux carried by the ELs, symmetries impose constraints on the distribution of these fluxes, and the conservation of topological flux (i.e., the source-free principle) ensures the formation of ECs. 

The conservation of braiding topological flux can even provide further insights. Breaking either one of the symmetries (say, $R_2$) causes the corresponding EL to deviate from the symmetry-invariant plane (see the blue EL in Fig. \ref{Fig-ECformation}(c)). However, as long as the in-plane red EL still reverses its orientation at a point, it must connect with the out-of-plane EL at this orientation-reversal point to ensure topological flux conservation. As a result, the EC persists even in the presence of only $R_1$ symmetry. If and only if both symmetries are broken, the EC can be untied into two separate ELs, as shown in Fig.~\ref{Fig-ECformation}(e). Therefore, without the knowledge of the orientations of ELs assigned by the eigenvalue braidings, it is impossible to explain why the EC remains stable in the scenario of Fig.~\ref{Fig-ECformation}(c), and why the chain of ELs cannot be separated or transformed into an exceptional link, as illustrated in Fig. \ref{Fig-ECformation}(d), but can solely evolve into the two configurations in Fig.~\ref{Fig-ECformation}(e).

\begin{figure*}[t!]
\includegraphics[width=0.9\textwidth]{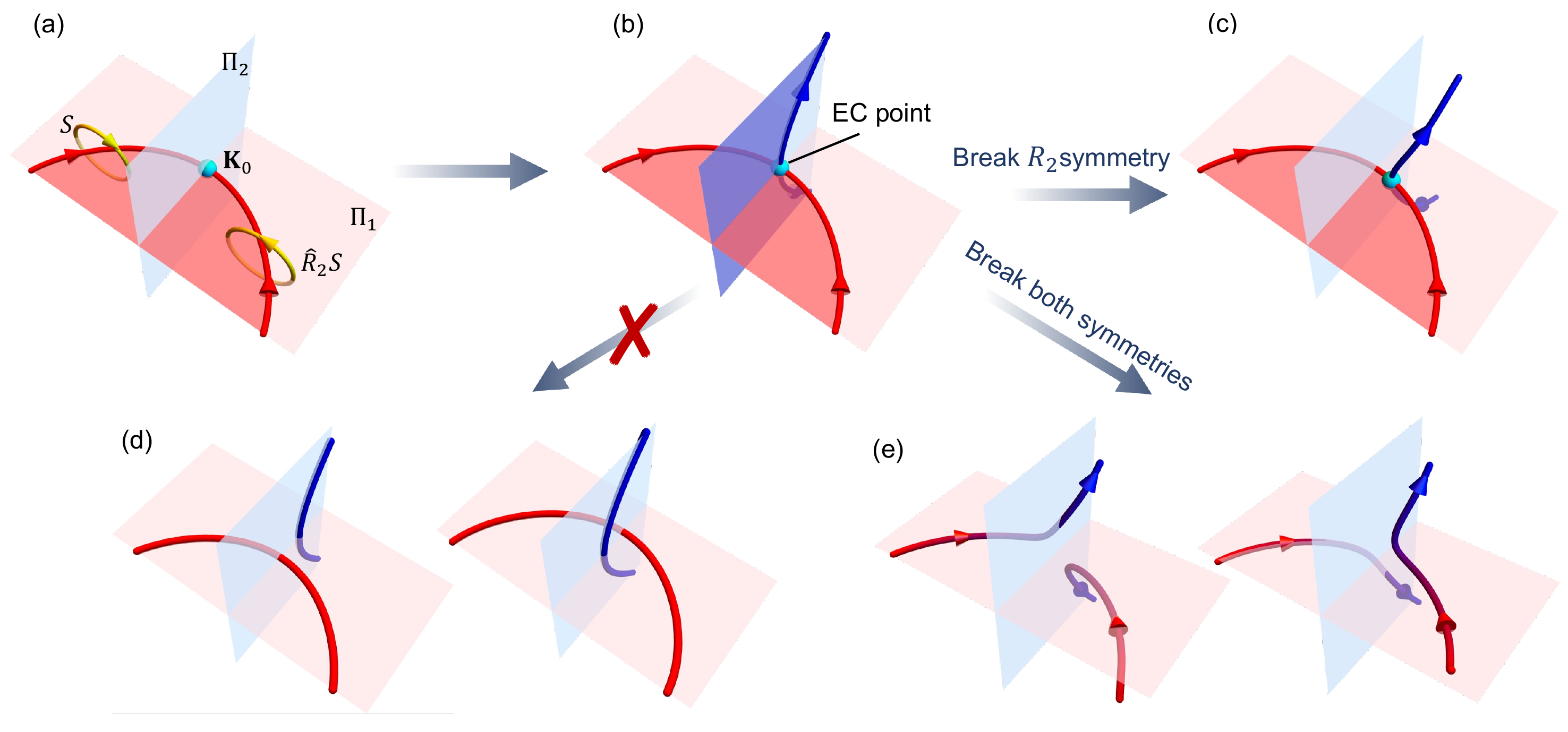}
\caption{\label{Fig-ECformation}
(a) An EL is confined by the symmetry $R_1$ in the symmetry-invariant plane $\Pi_1$. (B) Source-free principle of PF ELs demands the EL on $\Pi_2$ emerging from the chain point. (c) Beaking a single symmetry cannot lift the EC point. (d) Two EL configurations forbidden by the source-free principle. (e) Two possible evolutions when both symmetries are broken.
}
\end{figure*}

\subsection{Topological invariant characterizing symmetry-protected EC points}
As a new type of symmetry-protected topological (SPT) gapless phases, the EC points are stable against any symmetry-preserving perturbations and can only be eliminated by breaking the symmetries or through annihilation with other chain points.
Notwithstanding, the topology of EC points cannot be fully characterized by conventional methods, such as the enumeration of topological charges carried by loops or closed surfaces that encompass the singularities. The former is rendered ineffective because no loop can enclose a point in three dimensions, and the latter is inadequate since the surfaces wrapping EC points must undergo gap closing at the intersection points with ELs. 
Here, we overcome this problem by proposing a topological invariant carried by an open arc that terminates on the high-symmetry line to characterize the distinct topology of the symmetry-protected EC points.

\begin{figure*}[t!]
\includegraphics[width=0.55\textwidth]{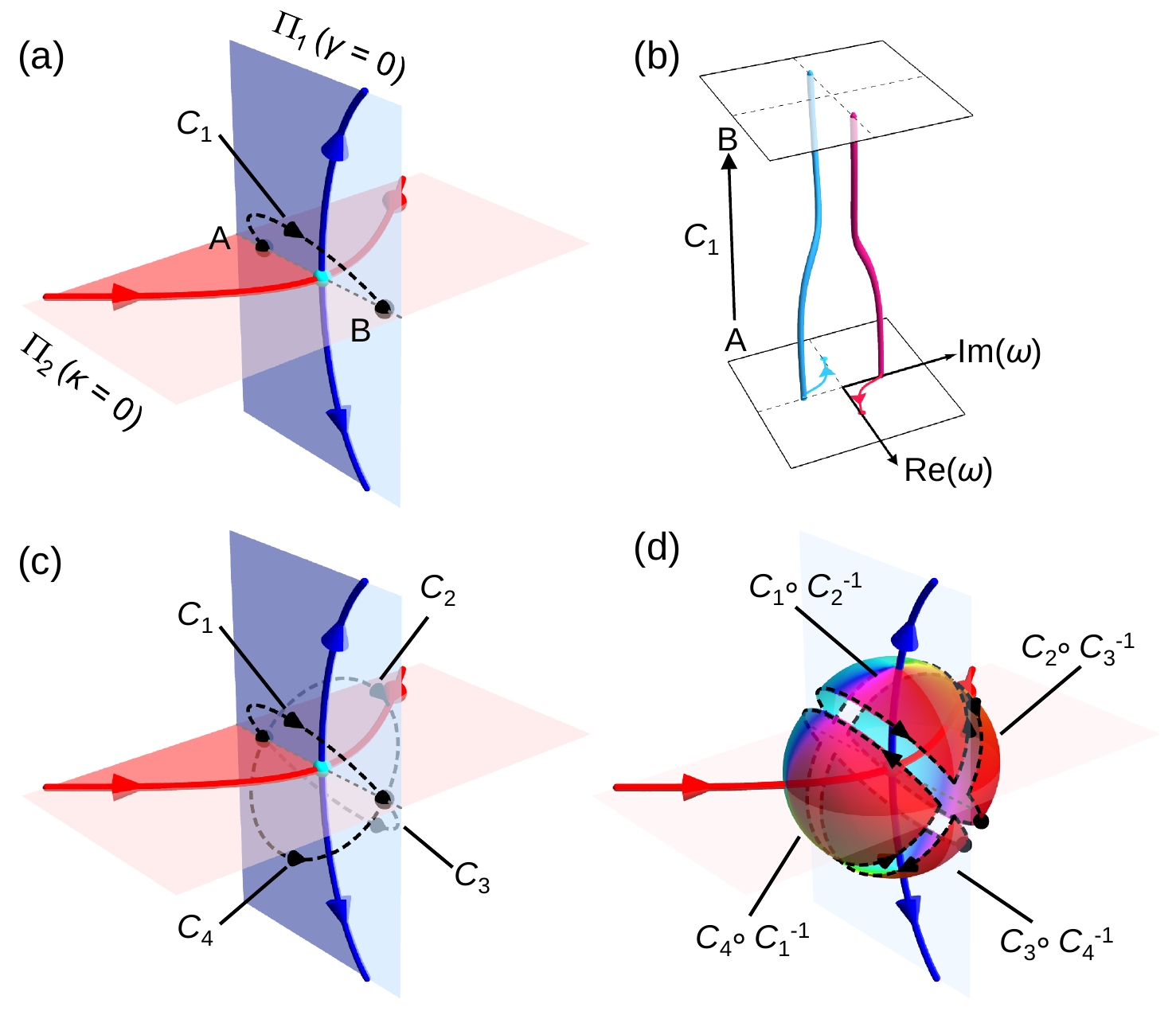}
\caption{\label{Fig-ChainInvariant}
(a) An open arc $C_1$ connecting two points on the intersection line of two symmetry-invariant planes, $\gamma=\kappa=0$.
(b) The quantized half-braid of two PF eigenfrequeices along the arc $C_1$. Since the initial and final eigenfrequencies are fixed along the real axis and a line parallel to the imaginary axis, respectively, the energy vorticity along $C_1$ must be a multiple of $1/4$.
(c) The three mirror images $C_2,C_3,C_4$ of the arc $C_1$ about the two planes $\Pi_1$ and $\Pi_2$. (d) The PF discriminant numbers carried by the four concatenated loops determine the number of ELs emanating from the chain points on the line segment $\overline{\mathrm{AB}}$. The four loops divide a closed surface enclosing the line segment $\overline{\mathrm{AB}}$ into four parts. The colors on the surfaces represent the argument of the PF discriminant $\arg\Delta_+$.  }
\end{figure*}

We consider the symmetry-protected ECs in the presence of both $\gamma$- and $\kappa$-symmetries introduced in the main text.
Since the EC points are confined along the high-symmetry line $\gamma=\kappa=0$ by the two symmetries, we may consider an open arc $C_1$ that terminates on this line, as shown in Fig. \ref{Fig-ChainInvariant}(a). As the eigenfrequencies are either real or in complex-conjugate pairs on this line, the PF discriminant takes real values along this line:
\begin{equation}
    \Delta_{+}(\mathbf{g})=\prod_{0<m<n}\left(\omega_m(\mathbf{g})-\omega_n(\mathbf{g})\right)^2 \in \mathbb{R}, \quad \forall \mathbf{g}=(\gamma, \chi, \kappa)=(0, \chi, 0).
\end{equation}
Therefore, as long as the eigenvalues at the endpoints A, B of the arc are non-degenerate, $\Delta_+(\mathrm{A})$, $\Delta_+(\mathrm{B})$ are non-zero, and the PF discriminant number along the open arc $C_1$ is half-quantized (the corresponding energy vorticity is quarter-quantized):
\begin{equation}
    \mathcal{D}_{+}\left(C_1\right)=\frac{1}{2 \pi} \int_C d \mathbf{g} \cdot \nabla_{\mathbf{g}} \arg \Delta_{+}(\mathbf{g})=\frac{1}{2 \pi}\left(\arg \Delta_{+}(\mathrm{B})-\arg \Delta_{+}(\mathrm{A})+2 n^{\prime} \pi\right)=\frac{n}{2} \in \mathbb{Z} / 2,
\end{equation}
which means the eigenfrequencies along the arc can precisely braid a multiple of half turns (see, e.g., Fig. \ref{Fig-ChainInvariant}(b)). Moreover, by taking the mirror images of $C_1$ about the $\gamma=0$ and $\kappa=0$ planes, we obtain another three image arcs $C_2, C_3, C_4$ as illustrated in Fig. \ref{Fig-ChainInvariant}(c), and the PF discriminant numbers along these arcs can be directly obtained according to the spectral symmetry:
\begin{equation}
\mathcal{D}_{+}\left(C_2\right)=\mathcal{D}_{+}\left(C_4\right)=-\mathcal{D}_{+}\left(C_3\right)=-\mathcal{D}_{+}\left(C_1\right).
\end{equation}
By concatenating the arcs into four loops, we observe that the PF discriminant numbers carried by the loops represent the net number of ELs outgoing from the line segment $\overline{\mathrm{AB}}$ on the four half-planes, respectively:
\begin{equation}
    \begin{aligned} & \mathcal{D}_{+}\left(C_1 \circ C_2^{-1}\right)=\mathcal{D}_{+}\left(C_3 \circ C_4^{-1}\right)=2 \mathcal{D}_{+}\left(C_1\right)=n, \\ & \mathcal{D}_{+}\left(C_2 \circ C_3^{-1}\right)=\mathcal{D}_{+}\left(C_4 \circ C_1^{-1}\right)=-2 \mathcal{D}_{+}\left(C_1\right)=-n,\end{aligned}
\end{equation}
and hence give the net number of EC points between the two points A,B along the high-symmetry line. 
Consequently, we find that the PF discriminant number $D_+(C_1)$ along the open arc $C_1$ is a quantized topological invariant associated with the symmetry-protected EC points, and twice the value of $D_+(C_1)$ determines the net number of EC points, $n = |2D_+(C_1)|$, situated between the two endpoints of the arc $C_1$ on the high-symmetry line $\gamma=\kappa=0$. If the arc only encompasses a single chain point, $D_+(C_1) = \pm 1/2$ and the sign determines the chirality of the chain point, which manifests as whether the ELs on the plane $\Pi_1$ emanate outward or inward from the chain point. Additionally, if either the $\gamma$-symmetry or $\kappa$-symmetry is broken, $D_+(C_1)$ is no longer quantized, confirming that $D_+(C_1)$ is a topological invariant exclusively associated with the ECs protected by these two symmetries.

\section{Symmetry correspondence between the quadratic and linearized eigenvalue problems}
For two coupled oscillators [Eq.~(3) of the main text], the linearized Hamiltonian is 
\begin{equation}\label{Hamiltonian}
\mathcal{H} = i\left( {\begin{array}{*{20}{c}}
0&{  \mathbf{1}}\\
{-\mathbf{M}^{-1}\mathbf{K}}&{{-\mathbf{M}^{-1}\mathbf{\Gamma }}}
\end{array}} \right),
\end{equation}
In the main text, we have shown that the QMP of the designed mechanical oscillators possesses two symmetries in the synthetic space spanned by the 3D vectors $\mathbf{g}=(\gamma,\chi,\kappa)$. In this section, we will discuss their correspondence to the symmetries of the linearized Hamiltonian $\mathcal{H}$ in the synthetic space. 

\subsection{Symmetry corresponding to the $\gamma$-symmetry of the QMP}
We first study the symmetry of the linearized Hamiltonian corresponding to the $\gamma$-symmetry of the QMP [Eq.~(4) of the main text]. The $\gamma$-symmetry is an antiunitary symmetry for the QMP. In general, an antiunitary operator $A=U\mathcal{K}\, (\omega\rightarrow\omega^*)$ ($U$ is the unitary part, $\mathcal{K}$ denotes complex conjugation) acting on the eigen vibration mode $\ket{\psi}=(a_1,\cdots,a_N)^\intercal$ can be expressed as $\ket{\psi}'=U\ket{\psi}^*$. Meanwhile, since the eigenstates for the linearized Hamiltonian is $\ket{\Psi}=(\ket{\psi},-i\omega\ket{\psi})^\intercal$, the antiunitary operation on $\ket{\Psi}$ manifests as
\begin{equation}
    \ket{\Psi}'=\begin{pmatrix}
    U\ket{\psi}^* \\ -i\omega^*U\ket{\psi}^*\end{pmatrix}=
    \begin{pmatrix} U & 0 \\ 0 & -U \end{pmatrix}
    \begin{pmatrix}
    \ket{\psi} \\ -i\omega\ket{\psi}\end{pmatrix}^*
    =\qty[\tau_z\otimes U]\mathcal{K}\ket{\Psi}.
\end{equation}
Therefore, we obtain the correspondence between the antiunitary symmetries of the QMP $Q(\omega)$ and the linearized Hamiltonian $\mathcal{H}$:
\begin{equation}
    UQ(\omega^*)^*U^{-1}=Q(\omega)\quad\Leftrightarrow\quad
    [\tau_z\otimes U]\mathcal{H}^*[\tau_z\otimes U]^{-1}=\mathcal{H}.
\end{equation}
For the $\gamma$-symmetry of our theoretical model, $Q(\omega^*,-\gamma)^*=Q(\omega,\gamma)$, the antiunitary operator is given by $A=\sigma_0(\gamma\rightarrow-\gamma,\omega\rightarrow\omega^*)\mathcal{K}$, namely the unitary part is $U=\sigma_0(\gamma\rightarrow-\gamma)$,  so we obtain the corresponding antiunitary symmetry for the linearized Hamiltonian: 
\begin{equation}
\left[\tau_z \otimes \sigma_0  \right] \mathcal{H}^*(\gamma) \left[\tau_z \otimes \sigma_0  \right]
=- i\left( {\begin{array}{*{20}{c}}
0&{ - \mathbf{1}}\\
{\mathbf{K}}&{{-\mathbf{\Gamma }}\left( {  \gamma } \right)}
\end{array}} \right) 
= {\mathcal H}\left( { -\gamma } \right).
\end{equation}

\subsection{Non-Hermitian latent symmetry corresponding to the $\kappa$-symmetry of the QMP}
Unlike the $\gamma$-symmetry for the QMP, the $\kappa$-symmetry [Eq.~(5) in the main text] cannot be reexpressed as a usual symmetry operation in the form $L\mathcal{H}L^{-1}=\mathcal{H}\ \text{or}\ \mathcal{H}^\dagger$ for the linearized Hamiltonian, which implies that some symmetry properties of mechanical systems would be more easily identified in the original QMP form than in the linearized form. 

Nevertheless, recent studies revealed that some spectral degeneracy of a linear Hamiltonian can be attributed to latent symmetries of the Hamiltonian~\cite{rontgen2021Latent,morfonios2021Flat}, going beyond the usual symmetries described by (anti)unitary operators that commute with the Hamiltonian. Here, we will extend the concept of latent symmetries to the non-Hermitian scenario and show that the $\kappa$-symmetry of the mechanical quadratic vibration equation is indeed equivalent to such a non-Hermitian latent symmetry of the linearized Hamiltonian.

\begin{definition}[Non-Hermitian latent symmetry]
For a non-Hermitian $N$-dimensional Hamiltonian $\mathcal{H}$, the degrees of freedom spanning the Hamiltonian are divided into two sets, i.e., the contributing set $S$ and the complement set $\bar{S}$.
A non-Hermitian latent symmetry of the Hamiltonian indicates that 
\begin{equation}
    \label{Eq-latentHert}L\left(\mathcal{H}^n\right)_{SS} L^{-1}=\left(\mathcal{H}^n\right)_{SS}^\dagger,\qquad \forall\  n\in\mathbb{N}
\end{equation}
where $L$ denotes the corresponding symmetry operator, and $\left(\mathcal{H}^n\right)_{SS}$ denotes the submatrix of $\mathcal{H}^n$ obtained by projecting $\mathcal{H}^n$ ($\mathcal{H}$ raised to $n$-th power) into the space of the contributing set of degrees of freedom.
\end{definition}
\vspace{8pt}
And according to the Cayley–Hamilton theorem, for any $n\geq N$, $\mathcal{H}^n$ can always be expressed as the matrix polynomial of $\mathcal{H}$ with lower powers $n<N$. Therefore, as long as $L\left(\mathcal{H}^n\right)_{SS} L^{-1}=\left(\mathcal{H}^n\right)_{SS}^\dagger$ holds for all $n<N$, the non-Hermitian latent symmetry is ensured.

On the other hand, we may define the isospectral reduction of a general non-Hermitian Hamiltonian~\cite{rontgen2021Latent,morfonios2021Flat}:
\vspace{8pt}
\begin{definition}[Isospectral reduction]
For a non-Hermitian Hamiltonian $\mathcal{H}=\begin{pmatrix}\mathcal{H}_{\bar{S}\bar{S}}&\mathcal{H}_{\bar{S}S}\\\mathcal{H}_{S\bar{S}}&\mathcal{H}_{SS}\end{pmatrix}$, the isospectral reduction of $\mathcal{H}$ over the contributing set $S$ of degrees of freedom is a matrix function of the eigenvalue $\omega$:
\begin{equation}
    \mathcal{R}_S(\mathcal{H},\omega)=\mathcal{H}_{SS}-\mH_{S\bar{S}}(\mH_{\bar{S}\bar{S}}-\omega I)^{-1}\mH_{\bar{S}S}.
\end{equation}
\end{definition}
\vspace{8pt}
The eigenvalues of the Hamiltonian satisfy
\begin{equation*}
    0=\det\qty[\mathcal{H}-\omega I]=\det\begin{pmatrix}\mathcal{H}_{\bar{S}\bar{S}}-\omega&\mathcal{H}_{\bar{S}S}\\\mathcal{H}_{S\bar{S}}&\mathcal{H}_{SS}-\omega\end{pmatrix}=\det\qty(\mH_{\bar{S}\bar{S}}-\omega)\det\qty[\mathcal{R}_S(\mathcal{H},\omega)-\omega].
\end{equation*}
Therefore, \textbf{provided that $\mH$ and $\mH_{\bar{S}\bar{S}}$ do not have common eigenvalues, the spectrum of $\mH$  coincides with that of the nonlinear eigen-equation for $\mathcal{R}_S(\mathcal{H},\omega)$:}
\begin{equation}\label{nonlinear eigeneq}
    \mathcal{R}_S(\mathcal{H},\omega)\ket{\psi}=\omega\ket{\psi}.
\end{equation}
This motivates calling $\mathcal{R}_S(\mathcal{H},\omega)$ an isospectral reduction of $\mH$. For the linearized Hamiltonian~\eqref{Hamiltonian} of mechanical oscillators, the isospectral reduction over the subspace of velocity $\ket{v}$ is 
$
    \mathcal{R}_{\ket{v}}(\mathcal{H},\omega)=-\mathbf{M}^{-1}\mathbf{\Gamma}+\frac{1}{\omega}\mathbf{M}^{-1}\mathbf{K}
$. Hence, the nonlinear eigen-equation~\eqref{nonlinear eigeneq} for $\mathcal{R}_{\ket{v}}(\mathcal{H},\omega)$ precisely gives the original quadratic eigen-equation of the oscillators:
\begin{equation}
    \qty[\mathcal{R}_{\ket{v}}(\mathcal{H},\omega)-\omega]\ket{\psi}=-\frac{1}{\omega}\mathbf{M}^{-1}(\omega\mathbf{\Gamma}-\mathbf{K}+\omega^2\mathbf{M})\ket{\psi}=0.
\end{equation}

It has been proved that the latent symmetry of a Hermitian Hamiltonian ~\eqref{Eq-latentHert} is equivalent to a symmetry of the isospectral reduction, $L\mathcal{R}_S(\mathcal{H},\omega)L^{-1}=\mathcal{R}_S(\mathcal{H},\omega)$~\cite{rontgen2021Latent,morfonios2021Flat}.
Here, as a generalization of the relation between the latent symmetry of the Hamiltonian and the symmetry of the isospectral reduction in the Hermitian case, we have the following theorem (see proof in Appendix):
\vspace{8pt}
\begin{theorem}[Alternative expression of non-Hermitian latent symmetry]\label{latent symmetry}
A non-Hermitian latent symmetry of the Hamiltonian $\mH$ over a contributing set $S$ is equivalent to a non-Hermitian symmetry of the isospectral reduction $\mathcal{R}_S(\mathcal{H},\omega)$ of $\mH$ over $S$:
\begin{equation}
    L\left(\mathcal{H}^n\right)_{SS} L^{-1}=\left(\mathcal{H}^n\right)_{SS}^\dagger,\quad \forall\  n\in\mathbb{N}\quad\Leftrightarrow\quad
    L\mathcal{R}_S(\mathcal{H},\omega)L^{-1}=\mathcal{R}_S(\mathcal{H},\omega^*)^\dagger,
\end{equation}
where $\mathbb{N}=\{1,2,3,\cdots\}$ denotes the set of natural numbers.
\end{theorem}

\vspace{8pt}
For our theoretical model of the mechanical oscillators, we see that the $\kappa$-symmetry of the QMP $Q(\omega)=\omega^2\mathbf{M}-\mathbf{K}+i\omega\mathbf{\Gamma}$, i.e., $\sigma_xQ^\dagger(\omega,\kappa)\sigma_x=Q(\omega^*,-\kappa)$, serves as a non-Hermitian symmetry of the isospectral reduction $\mathcal{R}_{\ket{v}}(\mathcal{H},\omega)$ with the symmetry operator $L=\sigma_x (\kappa\rightarrow-\kappa)$. Therefore, according to Theorem~\ref{latent symmetry}, the $\kappa$-symmetry also indicates a \textbf{latent $\kappa$-reflection-adjoint symmetry} of the linearized Hamiltonian~\eqref{Hamiltonian}:
\begin{equation}\label{latent kappa-symmetry}
    \sigma_x(\mathcal{H}(-\kappa)^n)_{BR}\sigma_x=(\mathcal{H}(\kappa)^n)_{BR}^\dagger,\quad\forall\ n\in\mathbb{N},
\end{equation}
where $(\mathcal{H}^n)_{BR}$ denotes the bottom right submatrix of the 2-by-2 block matrix $\mathcal{H}^n$. The existence of this latent symmetry of $\mathcal{H}$ can be directly checked. From Eq.~\eqref{Hamiltonian} and Eq.~(3) in the main text, we obtain
\begin{equation}
\begin{aligned}    
    \mathcal{H}(\kappa)_{BR} &=-i\mathbf{M}^{-1}\mathbf{\Gamma}=-i\frac{\gamma}{2m_0}\sigma_z,\\
    (\mathcal{H}(\kappa)^2)_{BR} &=\mathbf{M}^{-1}\mathbf{K}-(\mathbf{M}^{-1}\mathbf{\Gamma})^2=\qty(\frac{\bar\kappa}{m_0}-\frac{\gamma^2}{4m_0^2})\sigma_0-\frac{\chi+\delta\chi/2}{m_0}\sigma_x+i\frac{\delta\chi}{2m_0}\sigma_y+\frac{\kappa}{2m_0}\sigma_z,\\
    (\mathcal{H}(\kappa)^3)_{BR} &=i\qty[\mathbf{M}^{-1}\mathbf{\Gamma}\mathbf{M}^{-1}\mathbf{K}+\mathbf{M}^{-1}\mathbf{K}\mathbf{M}^{-1}\mathbf{\Gamma}+(\mathbf{M}^{-1}\mathbf{\Gamma})^3]=i\qty[\frac{\kappa\gamma}{2m_0^2}\sigma_0+\qty(\frac{\bar\kappa\gamma}{m_0^2}+\frac{\gamma^3}{8m_0^3})\sigma_z],
\end{aligned}
\end{equation}
which all satisfy Eq.~\eqref{latent kappa-symmetry} and further guarantee that every higher-order term $(\mathcal{H}^n)_{BR}$ ($n>3$) also satisfies Eq.~\eqref{latent kappa-symmetry} thanks to the Cayley-Hamilton theorem, thereby confirming the existence of the latent symmetry.

As we have shown in the main text, on the high-symmetry plane $\kappa=0$, the $\kappa$-symmetry ensures that the eigenfrequencies are either real or form complex conjugate pairs, corresponding to the $\kappa$-exact and $\kappa$-broken phases, respectively. Consequently, the ELs are fixed on the $\kappa=0$ plane, forming the phase transition boundaries of the exact and broken phases on the plane.




\section{The effective two-band Hamiltonian in quasi-degenerate approximation}
In this section, we will prove that in the quasi-degenerate approximation, the 2-by-2 QEP can be reduced to an effective two-band tight-binding model, and the $\gamma-$ and $\kappa-$symmetries of the QMP can be mapped to non-Hermitian space group symmetries in crystals.
Consider two coupled oscillators with equal masses $m_1=m_2=m_0$. When the two diagonal components of $\mathbf{K}$, $K_{11}$, $K_{22}$ are close enough and all other parameters in $\mathbf{K}$ and  $\mathbf{\Gamma}$ are sufficiently small, the two positive eigenfrequencies of the system will be quasi-degenerate $\omega=\omega_0+\delta\omega$ ($\delta\omega\ll\omega_0$) near the positive characteristic frequency $\omega_0= \sqrt{\frac{K_{11}+K_{22}}{2m_0}}$. Thus, the QEP can be approximated as a linear eigenvalue equation by omitting all the higher-order terms of $\delta\omega$: 
\begin{equation}\label{eq-H2reduced}
\begin{aligned}
{m_0\left( {{\omega _0} + \delta \omega } \right)^2}\left | \psi \right \rangle   -\mathbf K \left | \psi \right \rangle   + i({\omega_0} + \delta \omega )\mathbf \Gamma \left | \psi \right \rangle &= 0\\
m_0\left( {{\omega _0}^2 + 2{\omega _0}\delta \omega } \right)\left | \psi \right \rangle   - \mathbf K \left | \psi \right \rangle   + i({\omega _0} + \delta \omega ) \mathbf \Gamma \left | \psi \right \rangle   &= 0\\
\left( {2{\omega _0}m_0+ i{\bf{\Gamma }}} \right)\delta \omega \left | \psi \right \rangle   + \left( {{\omega _0}^2m_0 - {\bf{K}} + i{\omega _0}{\bf{\Gamma }}} \right)\left | \psi \right \rangle   &= 0\\
{\left( {2{\omega _0}m_0+ i{\bf{\Gamma }}} \right)^{ - 1}}\left[ {\left( {2{\omega _0}{ m_0 } + i{\bf{\Gamma }}} \right){\omega _0} - {\bf{K}} - {\omega _0}^2{ m_0 }} \right]\left | \psi \right \rangle   + \delta \omega \left | \psi \right \rangle   &= 0\\
{\left( {2{\omega _0} m_0+ i{\bf{\Gamma }}} \right)^{ - 1}}\left( {  {\bf{K}} + {\omega _0}^2m_0} \right)\left | \psi \right \rangle -(\omega_0+\delta \omega) \left | \psi \right \rangle &=0.
\end{aligned}
\end{equation}

Since we have assumed that $\mathbf{
\Gamma}$ is small, i.e., every element of the matrix $\mathbf{\Gamma}$ satisfies $ {2{\omega _0}{m_0} \gg {|\bf{\Gamma }}}_{mn}|$, Eq.~\eqref{eq-H2reduced} can be simplified as 
\begin{equation}
\begin{aligned}
{{\left( {2{\omega _0} m_0+ i{\bf{\Gamma }}} \right)}^{ - 1}}\left( {{\bf{K}} + {\omega _0}^2m_0} \right){\left | \psi \right \rangle}   &= \omega {\left | \psi \right \rangle} \\
{{\left( {1 + \frac{{i{\bf{\Gamma }}}}{{2{\omega _0}m_0}}} \right)}^{ - 1}}{{\left( {2{\omega _0}}m_0 \right)}^{ - 1}}\left( {{\bf{K}} + {\omega _0}^2m_0} \right){\left | \psi \right \rangle}   &= \omega {\left | \psi \right \rangle}\\
\left( {1 - \frac{{i{\bf{\Gamma }}}}{{2{\omega _0}m_0}}} \right)\left( {\frac{{\bf{K}}}{{2{\omega _0}m_0}} + \frac{{{\omega _0}}}{2}} \right){\left | \psi \right \rangle}   &= \omega {\left | \psi \right \rangle}\\
\left( {\frac{{{\omega _0}}}{2} - \frac{{i{\bf{\Gamma }}}}{4m_0} + \frac{{\bf{K}}}{{2{\omega _0}m_0}} - \frac{{i{\bf{\Gamma}\mathbf{K}}}}{{4{(\omega _0m_0)}^2}}} \right){\left | \psi \right \rangle}   &= \omega {\left | \psi \right \rangle}.
\end{aligned}
\end{equation}
We can further express $\mathbf{K}=m_0\omega_0^2 + \delta \mathbf{K}$, where every element of $\delta \mathbf{K}$ is a small number, i.e., $2m_0 \omega_0^2 \gg |\delta \mathbf{K}_{mn}|$, and by omitting the second-order small terms, the above equation becomes
\begin{equation}
\begin{aligned}
    &\left( {\omega_0 - \frac{{i{\bf{\Gamma }}}}{2m_0} + \frac{{\delta\bf{K}}}{{2{\omega _0}m_0}}} \right){\left | \psi \right \rangle} = \omega {\left | \psi \right \rangle}\\
        \Rightarrow \quad&
        \left( { \frac{{\delta\bf{K}}}{{2{\omega _0}m_0}} - \frac{{i{\bf{\Gamma }}}}{2m_0}} \right){\left | \psi \right \rangle} = (\omega-\omega_0) {\left | \psi \right \rangle}\\
         \Rightarrow  \quad&
         \mathcal{H}_\text{eff} \left| \psi \right \rangle = \frac{{1}}{{2{\omega _0}m_0}} \left( { \delta\mathbf{K} - i\omega_0\mathbf{\Gamma }} \right){\left | \psi \right \rangle} = \delta \omega \left| \psi \right \rangle.
\end{aligned} 
\end{equation}
Hence, we obtain the effective two-band linearized Hamiltonian for two generic coupled oscillators under the quasi-degenerate approximation:
\begin{equation}\label{effective hamiltonian}
    \mathcal{H}_\text{eff}=\frac{{1}}{{2{\omega _0}m_0}} \left( { \delta\mathbf{K} - i\omega_0\mathbf{\Gamma }} \right).
\end{equation}

Now, we consider the theoretical model [Eq.~(3) in the main text]. Assuming $\kappa$, $\chi$, $\gamma$, $\delta \chi$ are much smaller than $\bar\kappa$, we can choose  $\omega_0=\sqrt{\bar\kappa/m_0}$. Then, by substituting $\delta\mathbf{K} = \mathbf{K}-m_0\omega_0^2$ and $\mathbf{\Gamma}$ of Eq.~(3) in the main text into Eq.~\eqref{effective hamiltonian}, we can obtain the effective two-band Hamiltonian for the theoretical model as:
\begin{equation}
\begin{aligned}\label{theoretical effective}
\mathcal{H}_\text{eff-th} 
&= 
\frac{1}{2\omega_0m_0}\left( {\begin{array}{*{20}{c}}
\kappa/2-i\omega_0\gamma/2 &{{-\chi} }\\
{{-\chi} - \delta\chi }&-\kappa/2 +i\omega_0\gamma/2
\end{array}} \right) \\
&= \frac{1}{4\omega_0m_0}\left[\left( {\kappa  - i{\omega _0}\gamma } \right){\sigma_z} +(\delta\chi-2 {\chi}){\sigma_x} + i\delta\chi {\sigma_y}\right].
\end{aligned}
\end{equation}
Therefore, in the quasi-degenerate limit, the $\gamma$- and $\kappa$-symmetries of the QMP reduce to the following symmetries of $\mathcal{H}_\mathrm{eff-th}$:  
\begin{gather}
      \mathcal{H}_\text{eff-th}^* (\gamma) = \mathcal{H}_\text{eff-th} (-\gamma),\label{C2T}\\
        \sigma_x \mathcal{H}_\text{eff-th}^\dagger (\kappa)\sigma_x = \mathcal{H}_\text{eff-th} (-\kappa) \label{m-dagger}.
\end{gather}
If we map the synthetic space to the 3D momentum space with regarding the three synthetic parameters $(\gamma,\chi,\kappa)$ as the components of the momentum, \textbf{the reduced $\gamma$-symmetry~\eqref{C2T} manifests as a combined symmetry of  $C_2$-rotation about the $\gamma$ axis (i.e., $\kappa=\chi=0$) and time-reversal: $C_{2\gamma}T$,  in the momentum space. And the reduced $\kappa$-symmetry~\eqref{m-dagger} turns to be an effective non-Hermitian spatial symmetry, termed mirror-adjoint symmetry~\cite{zhang2022symmetry} about the $\kappa$ axis ($M_\kappa\hbox{-}\dagger$)}, which means after taking mirror reflection in the $\kappa$ direction, the effective Hamiltonian converts to its Hermitian-adjoint $\mathcal{H}_\mathrm{eff}^\dagger$. As we have shown in Ref.~\cite{zhang2022symmetry}, the presence of both $C_2T$ and mirror-adjoint symmetries about two orthogonal directions can protect the formation of an orthogonal exceptional chain, which also explains the formation of the exceptional chain in our theoretical mechanical model in a perturbative manner.


For the experimental model [Eq.~(6) of the main text], we can also obtain the effective two-band linearized Hamiltonian from Eq.~\eqref{effective hamiltonian} in the quasi-degenerate limit:
\begin{equation}
\begin{aligned}
\mathcal{H}_\text{eff-exp} 
&= 
\frac{1}{2\omega_0m_0}\left( {\begin{array}{*{20}{c}}
\chi-i\omega_0(\gamma_0+\gamma) &{{-\chi} }\\
{{-\chi} - \delta\chi }&-\kappa+\chi +i\omega_0\gamma_0
\end{array}} \right) \\
&= \mathcal{H}_\text{eff-th}+\frac{\chi-\kappa/2-i\omega_0(\gamma_0+\gamma/2)}{2\omega_0m_0}\sigma_0,
\end{aligned}
\end{equation}
where we have identified the characteristic frequencies of the two models, i.e., $\omega_0=\sqrt{\bar\kappa/m}=\sqrt{\kappa_0/m}$. We see that the effective two-level Hamiltonian for the experimental model is only different from the effective theoretical model~\eqref{theoretical effective} by a $\sigma_0$ term. Since the $\sigma_0$ term does not affect the difference between the two eigenvalues of the effective Hamiltonian, the degeneracies of the experimental model are identical to those of the theoretical model under the quasi-degenerate approximation. Nevertheless, we stress that the EC observed in the experiment is not an approximate result in the quasi-degenerate limit. As we analyzed in the main text (also see Section~\ref{subspace symmetries} in SM), the two subspace symmetries of the experimental model can exactly protect the EC to remain stable beyond the quasi-degenerate regime.

\section{Experimental setup and measurement}
In this section, we will show how to construct the experimental model [Eq.~(6) in the main text] and describe how to realize the nonreciprocity and controllable synthetic dimensions in experiments.

\subsection{Derivation of the experimental model}
\begin{figure*}[t!]
\includegraphics[width=0.9\textwidth]{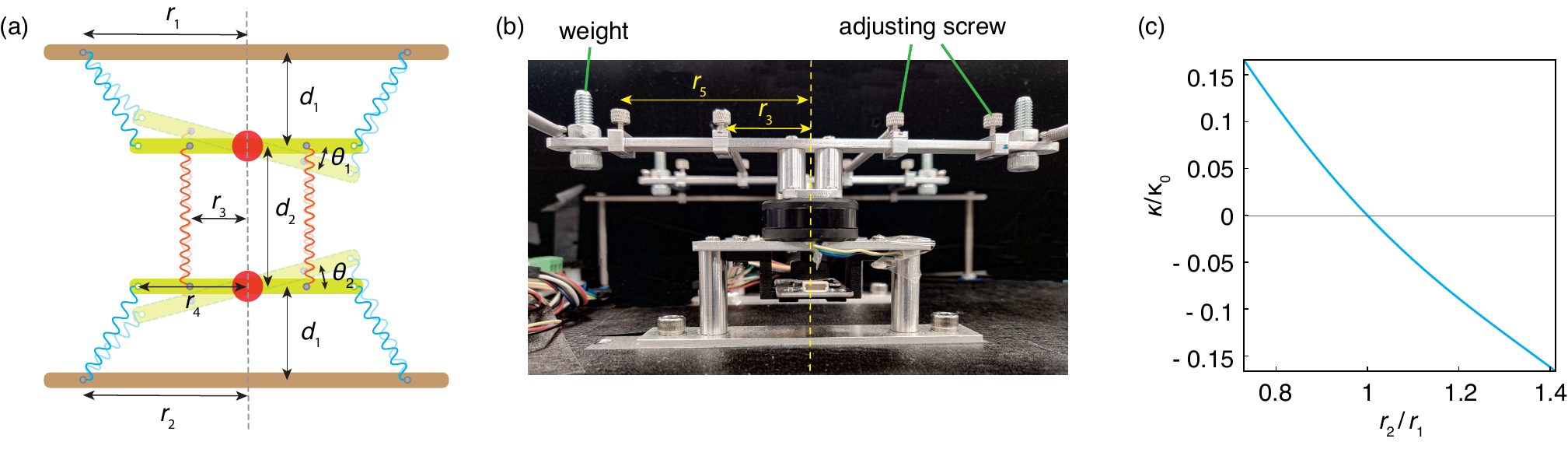}
\caption{\label{Fig-expSetup}
(a) The geometric parameters of the experimental setup. (b) The photo of the experimental setup. (c)
$\kappa/\kappa_0$ as a function of $r_2/r_1$ and the other geometric parameters used are: $r_1 = 0.085$m, $r_4 = 0.0865$m, $d_1=d_2=d_3=0.1$m, and $l_1=0.05$m.
}
\end{figure*}

Figure~\ref{Fig-expSetup} illustrates our experimental system composed of two coupled active rotational harmonic oscillators. A rotational arm (green) anchored on the motor (red) is attached to the rigid beam (brown) by two identical springs (cyan) and another two identical springs (orange) connect the two rotational arms in parallel.
The torques on the motors can be divided into a conservative part exerted by the springs connecting the arms and a non-conservative part induced by the inherent damping as well as by the computer-numerically-control Ampere force inside the motors. The conservative torques on the two motors can be directly obtained from the potential energies, $V_1$, $V_2$, $V_3$, stored in the two upper cyan springs, the two lower cyan ones, and the two middle orange ones, respectively, which can always be expressed as the following quadratic forms for small-angle vibrations $\theta_1,\theta_2\ll1$:
\begin{equation}
    \begin{aligned}
    V_1 &=\frac{1}{2}\kappa_0\theta_1^2+V_1^0,\\
    V_2 &=\frac{1}{2}(\kappa_0-\kappa)\theta_2^2+V_2^0,\\
    V_3 &=\frac{1}{2}\chi(\theta_1-\theta_2)^2+V_3^0,
    \end{aligned}
\end{equation}
where $V_1^0$, $V_2^0$, and $V_3^0$ represent the potential energy of the springs at the equilibrium position $\theta_1=\theta_2=0$. 
The absence of linear terms of $\theta_1$, $\theta_2$ in the above three potentials is due to the fact that the net torque of each group of springs equals zero at the equilibrium position, i.e., $\partial_{\theta_i} V_j\big|_{\theta_1=\theta_2=0}=0$. Consequently, the conservative torque on each motor can be obtained from the derivatives of the total potential energy $V(\theta_1,\theta_2)=V_1+V_2+V_3$:
\begin{equation}\label{conservative torque}
    \tau^\mathrm{c}_{1}=-\frac{\partial V}{\partial\theta_1}=-\kappa_0\theta_1-\chi(\theta_1-\theta_2),\qquad
    \tau^\mathrm{c}_{2}=-\frac{\partial V}{\partial\theta_2}=-(\kappa_0+\kappa)\theta_2+\chi(\theta_1-\theta_2).
\end{equation}

The nonreciprocal and velocity-dependent effects are induced by the non-conservative torques applied on the two motors,
\begin{equation} \label{EqS-nctorque}
    \tau_1^\mathrm{nc} = \gamma_1 \frac{d \theta_1}{dt}=(\gamma_0+\gamma/2) \frac{d \theta_1}{dt}, \quad
    \tau_2^\mathrm{nc} = \delta \chi \theta_1 + \gamma_2 \frac{d \theta_2}{dt} =\delta \chi \theta_1 + (\gamma_0-\gamma/2) \frac{d \theta_2}{dt},
\end{equation}
where $\delta \chi$ characterizes the strength of the nonreciprocal effect. The total velocity-dependent damping effects on motors 1 and 2 are $\gamma_{1,2}=\gamma^\mathrm{in}+\gamma^\mathrm{ex}_{1,2}$, where $\gamma^\mathrm{in}$ represents the inherent damping of the two motors and $\gamma_{1,2}^\mathrm{ex}$ represents the tunable damping controlled by the Ampere force. We denote $\gamma_{0}$ as the common loss of the two motors, and $\gamma$ represents the difference of Ampere force applied on these two motors. 
From Eqs.~(\ref{conservative torque},\ref{EqS-nctorque}), we obtain
the matrices in the SDE of the experimental system
\begin{equation}\label{Eq-ExperModelOrigin}
\mathbf K = \left( {\begin{array}{*{20}{c}}
{{\kappa_{0}+\chi}}&{{-\chi}}\\
{{-\chi-\delta\chi}}&{{\kappa_{0}-\kappa+\chi}}
\end{array}} \right), 
\quad  
\mathbf \Gamma =\left( {\begin{array}{*{20}{c}}
{{ \gamma_0 + \gamma/2}}&{0}\\
{0}&{{\gamma_0 - \gamma/2}}
\end{array}} \right).
\end{equation}

\subsection{Parameter control in the experiments}
The coefficients in the potential energy equations are determined by the geometric parameters of the experimental setup \ref{Fig-expSetup}. The potential energies can be expressed as
\begin{equation}
\begin{aligned}
    V_1 &= \frac{1}{2} k_1 \left[(\sqrt{(r_4 \cos\theta_1 - r_1)^2 + (r_4 \sin\theta_1 - d_1)^2} - l_1)^2 + (\sqrt{(-r4 \cos\theta_1 + r_1)^2 + (-r_4 \sin\theta_1 - d_1)^2} - l_1)^2\right],\\
    V_2 &= \frac{1}{2} k_1 \left[(\sqrt{(r_4 \cos\theta_2 - r_2)^2 + (r_4 \sin\theta_2 - d_1)^2} - l_1)^2 + (\sqrt{(-r4 \cos\theta_2 + r_2)^2 + (-r_4 \sin\theta_2 - d_1)^2} -l_1)^2\right],\\
    V_3 &= \frac{1}{2} k_2 \left[(\sqrt{r_3^2(\cos\theta_1-\cos\theta_2)^2 + (d_2+r_3\sin\theta_1-r_3\sin\theta_2)^2 } - l_2)^2 \right. \\
   &\qquad\qquad\left. +(\sqrt{r_3^2(\cos\theta_1-\cos\theta_2)^2 + (d_2-r_3\sin\theta_1+r_3\sin\theta_2)^2 } - l_2)^2\right],    
\end{aligned}
\end{equation}
where $k_{1,2}$ and $l_{1,2}$ denote the stiffness and free length of the cyan ($k_1$, $l_1$) and orange ($k_2$, $l_2$) springs. 
By taking derivatives of the above potentials and keeping the linear terms, we can obtain the coefficients as
\begin{equation}\label{EqS-parafun}
\begin{aligned}
    \chi &= 2k_2r_3^2,\\
    \kappa_0 &= 2 k_1 r_4 \left[r_1 - \frac{l_1 r_1}{\sqrt{d_1^2 + (r_1 - r_4)^2}} + \frac{d_1^2 l_1 r_4}{(d_1^2 + (r_1 - r_4)^2)^{3/2}}\right],\\
    \kappa_0-\kappa &= 2 k_1 r_4 \left[r_2 - \frac{l_1 r_2}{\sqrt{d_1^2 + (r_2 - r_4)^2}} + \frac{d_1^2 l_1 r_4}{((d_1^2 + (r_2 - r_4)^2)^{3/2}}\right].
\end{aligned}
\end{equation}
From the above equations, we see that $\chi$ is a function of $r_3$, and $\kappa$ can be controlled by fixing $r_1$ and tuning $r_2$. And Fig.~\ref{Fig-expSetup}(c) shows that  $\kappa$ decreases monotonically with  $r_2$ in the range $r_2/r_1 \in(0.8,1.4)$.

Therefore, in the experiment, we can control the value of $\kappa$ by tuning $r_2$. As shown in Fig.~\ref{Fig-expSetup}(b), the ends of each spring were connected with the arms and beams by two adjusting screws, ensuring that the connection positions $r_3$ and $r_2$ could be continuously tuned. 
Weights (bolts and nuts) are mounted on the ends of the rotational arm to control the moment of inertia of the motor and hence control the resonant frequency of each individual oscillator. Note that since the weights of the adjusting screws are not negligible, varying the connection position $r_3$ will affect the moment of inertia $m_0$. To keep the moment of inertia unchanged when tuning $r_3$, we added two more movable screws at position $r_5$ on each arm so that $(r_3^2 +r_5^2)$ remains a constant.  

While the coefficients associated with the conservative torques are determined by the geometric parameters, the coefficients associated with the non-conservative torques are controlled by the Ampere forces whose magnitudes are determined by the rotation angles or velocities of the motors.  
The instantaneous rotation angles $\theta_n(t)$ and angular velocities $\frac{d}{dt}\theta_n(t)$ of the two motors are measured by a sensor placed below the motors and collected by their microcontrollers. The microcontrollers can be programmed to output PWM signals to apply customized torques on the motors, effectively realizing the required nonreciprocity and damping effects. 

The nonreciprocal interaction between the two rotational arms can be realized by a control loop that measures the rotation angle of motor 1 in real-time and applies a nonconservative torque, $ \delta\chi \theta_1(t)$, on motor 2 based on the measured angle. 
In addition, to actively control the damping matrix, we 
design a control loop that can apply additional torques on the two motors proportional to their own velocities, $\gamma_n d\theta_n(t)/dt $. By tuning the coefficients of the two velocity-dependent torques, we can control the relative damping, $\gamma = \gamma_1 - \gamma_2$, as the third synthetic parameter dimension.  

\subsection{Retrieval of the parameters using Green's function method}
\begin{figure*}[h!]
\includegraphics[width=0.8\textwidth]{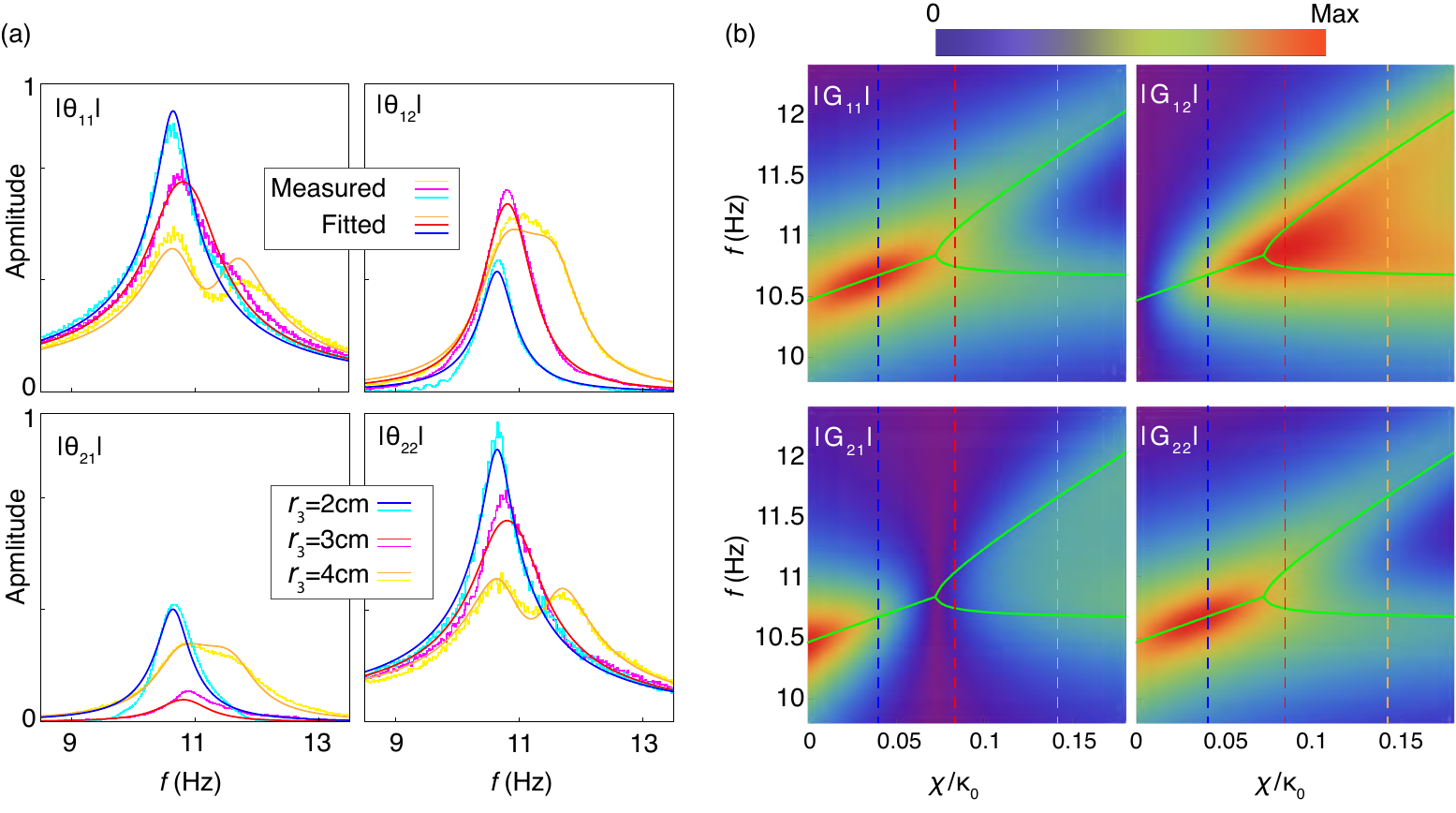}
\caption{\label{Fig-fitchi0} 
(a) The normalized 
measured response $|\theta_{mn}|$ at oscillator $m$ under the the excitation applied to oscillator $n$, and the fitted results for connection position $r_3=2,3,4$ cm. (b) The intensity graphs represent the Green's function distribution as functions of frequency $f=\omega/(2\pi)$ and coupling strength $\chi/\kappa_0$. The green line represents the real parts of the eigenfrequencies, and the dashed lines denote the three selected spectra in (a). 
The other fitted parameters used in the calculation are: $\kappa_0/m_0=4330$, $\gamma_0/\sqrt{m_0 \kappa_0}=0.085$, $\gamma=0$, $\kappa=0$, and $\delta\chi=-0.073 \kappa_0$.
}
\end{figure*}

\begin{figure}[h!]
\includegraphics[width=0.8\textwidth]{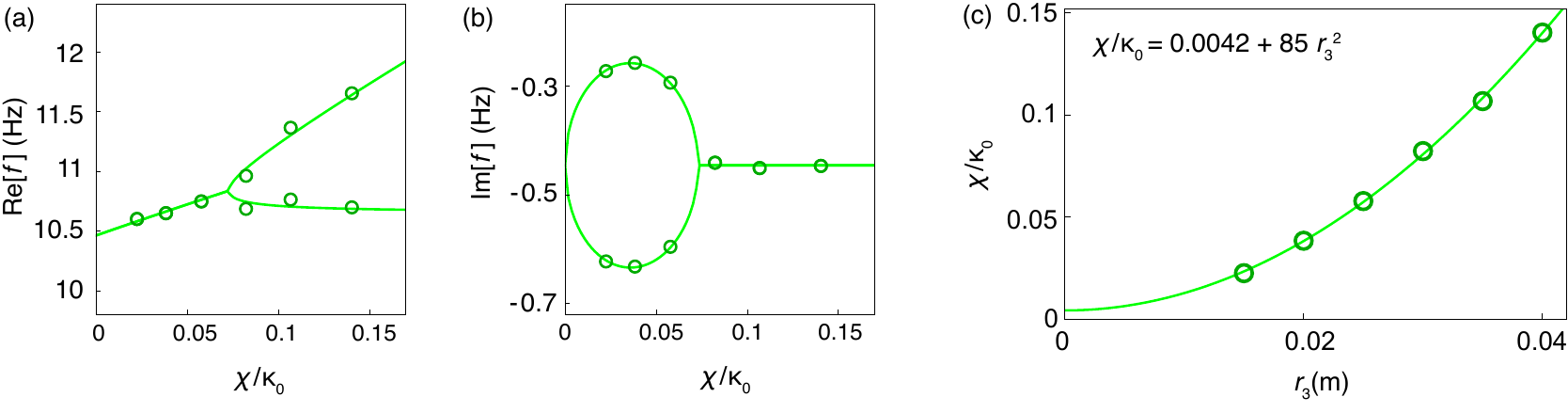}
\caption{\label{Fig-chi0function}
The measured (a) real and (b) imaginary parts of eigenfrequencies. 
The retrieved parameter $\chi$ is a parabolic function of the connection position $r_3$. }
\end{figure}

In the above section, we explained how to control the nonreciprocity $\delta\chi$ and synthetic parameters $\mathbf{g}=(\gamma,\chi,\kappa)$ in experiment, yet the values of these parameters still need to be retrieved by Green's function method. 
During experiments, we excited the motor $n (n=1,2)$ using a chirp signal covering a frequency range $2-22$ Hz and recorded the vibration amplitudes of oscillators 1 and 2 as $\theta_{1n} (t)$ and $ \theta_{2n} (t)$, respectively. After the Fourier transformation, we obtained the vibration amplitudes in the frequency domain. It is known that the linear response in the frequency domain can be expressed as 
\begin{equation}
    Q (\omega) |p_n \rangle = c |s_n \rangle
    \quad\Rightarrow\quad |p_n \rangle = c\, Q (\omega)^{-1}  |s_n \rangle,
\end{equation}
where $c$ is a fitting coefficient, $|p_n\rangle = (\theta_{1n},\theta_{2n})^\intercal=c(G_{1n}, G_{2n})^\intercal$ represents the probed vibration amplitudes of the two oscillators driven by the source $|s_n\rangle$ and $|s_1\rangle = (1,0)$ ($|s_2\rangle = (0,1)$) represents excitation applied on oscillator 1 (2). Therefore, the Green's function of the system (also known as transfer function for mechanical systems~\cite{tisseur2001Quadratic}) is just
\begin{equation}
    G(\omega)=Q(\omega)^{-1}=\qty[\omega^2\mathbf{M}  -  \mathbf{K} + i\omega \mathbf{\Gamma}]^{-1}.
\end{equation}
Consequently, by fitting the four elements, $G_{mn}(\omega)$, as functions of $\omega$ via the least-squares method, we can retrieve the parameters in $Q(\omega)$. 

Now let us explain how the band structures in Fig.~2 of the main text are systematically retrieved using the Green's function.  In order to precisely control the value of $\chi$ by varying $r_3$ in experiments, we first need to obtain the exact relation between $\chi$ and $r_3$.  
For this purpose, we first 
tuned the programmed damping effects applied on the two motors to be the same so that $\gamma=0$. We also fixed the positions of the springs attached to the beams at $r_1 =r_2 = 8.5$ cm so that $\kappa=0$. 
Then fixing $r_3$ at several different values, we excited the motors using chirp signals and recorded the four spectra $\theta_{mn}$ as shown in Fig.~\ref{Fig-fitchi0}(a). Using the expression of $|G_{mn}(\omega)|$ to fit the measured frequency spectrum $|\theta_{mn}(\omega)|$, we obtained the constant parameters in $Q(\omega)$: $\kappa_0/m_0=4330$, $\gamma_0/\sqrt{m_0 \kappa_0}=0.085$, and $\delta\chi=-0.073 \kappa_0$. And we also obtained
the values of the coupling strength $\chi$ for three different groups of curves:
$\chi(r_3=0.02)=0.038\kappa_0$, $\chi(r_3=0.03)=0.082\kappa_0$, and $\chi(r_3=0.04)=0.140\kappa_0$. As shown in Fig.~\ref{Fig-fitchi0}(a), the fitted spectra agree well with the measured spectra.   
Using the retrieved constant parameters, in Fig.~\ref{Fig-fitchi0}(b), we plot the linear response spectra as functions of frequency $f=\omega/(2\pi)$ and $\chi$, showing that the the loci of the resonances are in good agreement with the real parts of the eigenfrequencies.

We measured more response spectra for different $r_3$ and retrieved the values of $\chi$ using the Green's function method. Substituting the retrieved parameters into $Q(\omega)$, we obtained the measured eigenfrequencies, and plotted them in Fig.~\ref{Fig-chi0function}(a) and (b) by dots. Note that the dots deviate slightly from lines, because the green lines are plotted using a constant parameter $\kappa_0/m_0 = 4330$ while the dots are calculated by setting the $\kappa_0/m_0$ as a variable parameter in the fitting model to make the fitting spectra agree better with the measured spectra.
Then,  we fitted the retrieved coupling strength $\chi$ as a function of the connection position $r_3$ with the result shown in Fig.~\ref{Fig-chi0function}(c). We see that the fitted function is a parabola passing through the origin, consistent with our theoretical analysis \eqref{EqS-parafun}. 

Since we have obtained the precise relation between $\chi$ and $r_3$ in Fig.~\ref{Fig-chi0function}(c), $\chi$ can be regarded as a fixed number for certain $r_3$ when fitting different lines in Figs.~2(e) and (f) of the main text. In other words, the measured dots on different lines in Figs.~2(e) and (f) of the main text are retrieved by fixing $\delta \chi = -0.073 \kappa $, $\gamma_0/\sqrt{m_0 \kappa_0} = 0.085$, and $\chi/\kappa_0 = 0.042+ 85 {r_3}^2$.

\subsection{Derivation of the subspace symmetries of the experimental model}\label{subspace symmetries}
Owing to the intrinsic dissipation of the oscillators in experiments, here we show that the $\kappa$- and $\gamma$-symmetries of the theoretical model are reduced to two subspace symmetries. By shifting the complex frequencies by an imaginary displacement $i\omega_i$: $\omega=\tilde{\omega}+i\omega_i$, the QMP for the experimental system can be rewritten as
\begin{equation}\label{exp QMP}
    Q=\qty[\tilde{\omega}^2\mathbf{M}-(\mathbf{K}+\omega_i\tilde{\mathbf{\Gamma}}+\omega_i^2\mathbf{M}+\omega_i\gamma_0\sigma_0)+i\tilde{\omega}\tilde{\mathbf{\Gamma}}]+i\tilde{\omega}(2\omega_i m_0+\gamma_0)\sigma_0,
\end{equation}
with $\tilde{\mathbf{\Gamma}}=\mathbf{\Gamma}-\gamma_0\sigma_0=\sigma_z\gamma/2$. By setting $\omega_i=-\frac{\gamma_0}{2m}$, the last term in \eqref{exp QMP} vanishes and the QMP for the experimental model is altered to
\begin{equation}
    Q(\omega,\gamma,\chi,\kappa)=\tilde{Q}(\tilde{\omega},\gamma,\chi,\kappa)=\tilde{\omega}^2\mathbf{M}-\underbrace{\qty(\mathbf{K}-\frac{\gamma_0}{2m_0}\tilde{\mathbf{\Gamma}}-\frac{\gamma_0^2}{4m_0}\sigma_0)}_{\displaystyle \tilde{\mathbf{K}}}+i\tilde{\omega}\tilde{\mathbf{\Gamma}},
\end{equation}
which possesses the effectively loss-gain balance dissipation matrix $\tilde{\Gamma}=\sigma_z\gamma/2$ and the effective stiffness matrix
\begin{equation}
   \tilde{\mathbf{K}}=\begin{pmatrix}
        \tilde{\kappa}_0+\chi-\frac{\gamma_0}{4m_0}\gamma & -\chi \\
        -\chi-\delta\chi & \tilde{\kappa}_0-\kappa+\chi+\frac{\gamma_0}{4m_0}\gamma
    \end{pmatrix}
\end{equation}
with $\tilde{\kappa}_0=\kappa_0-\frac{\gamma_0^2}{4m_0}$. It is seen that $\tilde{Q}(\tilde{\omega},\gamma,\chi,\kappa)$ violates the $\kappa$- and $\gamma$-symmetries of the theoretical model. Nevertheless, $\tilde{Q}(\tilde{\omega},\gamma,\chi,\kappa)$ on two intersecting 2D planes are invariant under the following transformations, respectively:
\begin{align}
   \gamma=0\ \text{plane:}&\quad \tilde{Q}^*(\tilde{\omega},\gamma=0,\kappa,\chi)=\tilde{Q}(\tilde{\omega}^*,\gamma=0,\kappa,\chi)\\
   \kappa=\frac{\gamma_0}{2m_0}\gamma\ \text{plane:}&\quad \sigma_x\tilde{Q}^\dagger(\tilde{\omega},\kappa=\frac{\gamma_0}{2m_0}\gamma,\chi)\sigma_x=\tilde{Q}(\tilde{\omega}^*,\kappa=\frac{\gamma_0}{2m_0}\gamma,\chi).\label{subspace symmetry 2}
\end{align}
As a result, the complex eigenfrequencies on the two planes have spectral symmetry about the constant damping plane $\omega=i\omega_i=-i\frac{\gamma_0}{2m_0}$. In particular, akin to the $\kappa$-symmetry for the theoretical model, the subspace symmetry~\eqref{subspace symmetry 2} on the $\kappa=\frac{\gamma_0}{2m_0}\gamma$ plane can also be written as a non-Hermitian latent symmetry for the linearized Hamiltonian $\tilde{\mathcal{H}}=i\begin{pmatrix}
    0 & \mathbf{1}\\ -\mathbf{M}^{-1}\tilde{\mathbf{K}} & -\mathbf{M}^{-1}\tilde{\mathbf{\Gamma}}
\end{pmatrix}$ on that plane:
\begin{equation}\label{subspace latent kappa-symmetry}
    \sigma_x\qty(\tilde{\mathcal{H}}\qty(\kappa=\frac{\gamma_0}{2m_0}\gamma)^n)_{BR}\sigma_x=\qty(\mathcal{H}\qty(\kappa=\frac{\gamma_0}{2m_0}\gamma)^n)_{BR}^\dagger,\quad\forall\ n\in\mathbb{N}.
\end{equation}
Therefore, the emergence of the EC in the experimental system can be attributed to the subspace non-Hermitian latent symmetry of the model.




\section{Exceptional chain in 3D mechanical lattice and the effect of pulse splitting and needle pulse generation}\label{sec-3Dlattice}
In this section, we move on to a 3D active mechanical lattice and show that the non-Hermitian ECs can also stably exist in the momentum space of the 3D lattice under the protection of non-Hermitian latent crystalline symmetries unique to mechanical systems. Our finding implies that the notion of non-Hermitian latent symmetry is not only applicable to particular synthetic-dimension models, but is also fundamentally important in characterizing the non-Hermitian crystalline symmetries of general active mechanical lattices. Furthermore, we reveal that the peculiar complex eigenfrequency dispersion near the EC point can give rise to an extraordinary transport effect for rectangular wavepackets. Specifically, when the wavepacket scatters at the chain point, it will split into two pulses propagating in opposite directions; more surprisingly, the two pulses rapidly broaden in the transverse direction and evolve into needle pulses~\cite{wang2010Engineering,parker2016Longitudinal,grunwald2020Needle,cao2023Opticalresolution}. This exotic phenomenon demonstrates that, in addition to its theoretical significance as a new type of non-Hermitian topological nodal structure, the non-Hermitian ECs also have the potential to induce nontrivial physical effects and thus hold promising applications.

\subsection{Exceptional chain in 3D non-Hermitian nonreciprocal mechanical lattice}

\begin{figure*}[h]
\includegraphics[width=0.9\textwidth]{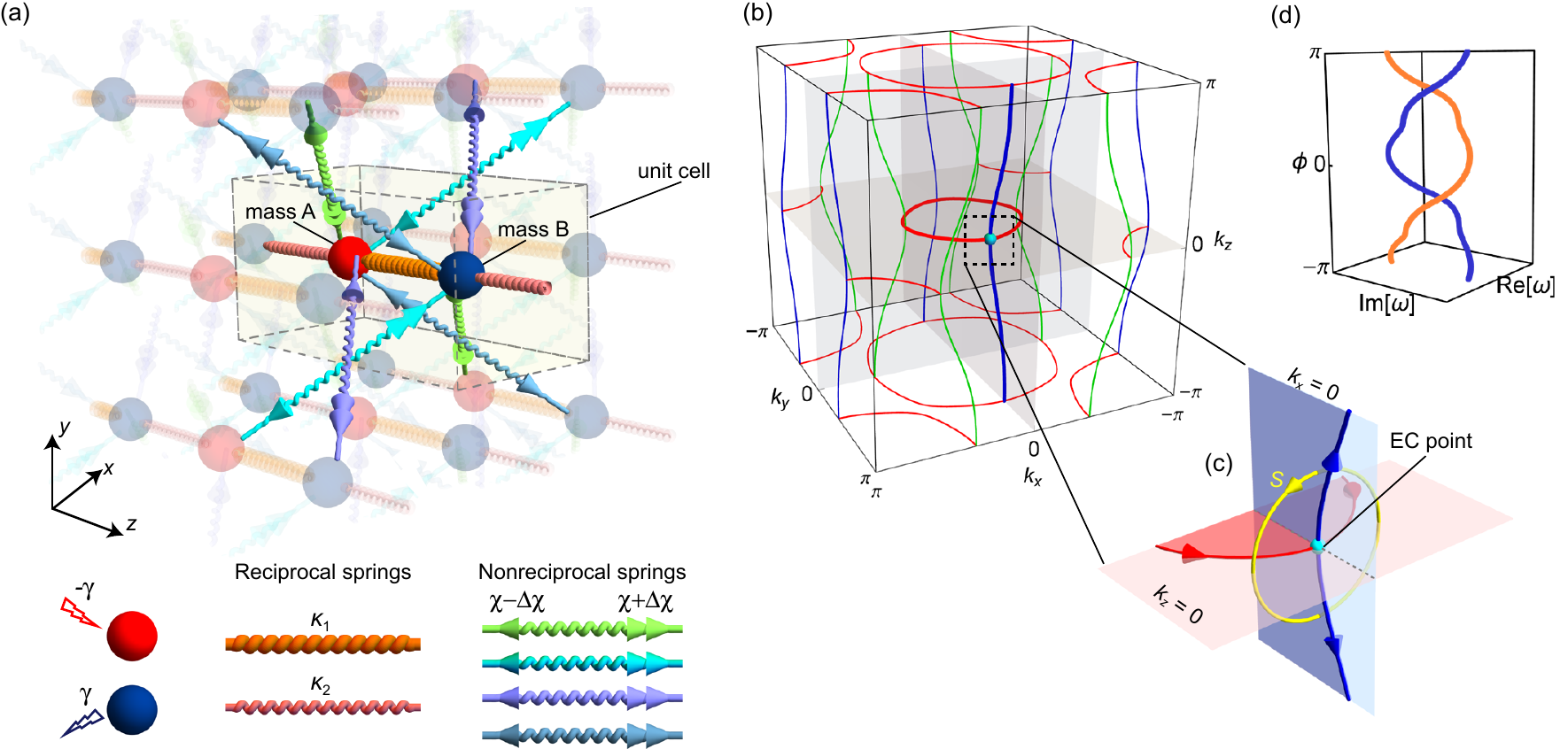}
\caption{\label{Fig-3Dlattice}
(a) The schematic of the 3D mechanical lattice. Two types of oscillators (red and blue) with intrinsic gain (anti-damping) and loss (damping) are connected by reciprocal and nonreciprocal springs. (b) The exceptional chain configuration in the momentum space under the parameters: $\kappa_1=1.3\kappa_0$, $\kappa_2=-0.7\kappa_0$, $\chi=0.5\kappa_0$, $\Delta\chi=0.4\kappa_0$, and $\gamma=0.7\sqrt{m\kappa_0}$. (c) The zoom-in plot of the nodal structure near an exceptional chain point. (d) The eigenfrequencies braiding along the loop $S$ in (c).
}
\end{figure*}

We consider a 3D mass-spring orthorhombic lattice shown in Fig.~\ref{Fig-3Dlattice}(a). In each unit cell (see, e.g., the marked unit cell in the dashed box), there exist two oscillators A and B that are fixed along the $z$ axis to undergo 1D vibrations. The two oscillators have the same mass $m$ and the same nature oscillating frequency $\sqrt{\kappa_0/m}$, but they have opposite intrinsic damping coefficients $-\gamma$ (i.e, gain for A) and $\gamma$ (loss for B), respectively. All springs connecting the oscillators in the marked unit cell are highlighted in solid colors, while other springs are displayed in translucency for clarity.  The intracell and intercell springs connecting the two types of oscillators in the $z$ directions are reciprocal with the stiffnesses $\kappa_1$ and $\kappa_2$, respectively. In contrast, the oblique springs connecting the next-nearest-neighbor unit cells are all nonreciprocal, which means that the restoring forces ($z$ component) exerted to the oscillators on the two sides take different values $f_\mathrm{da}=-(\chi+\Delta\chi) \delta z$ (the double-arrow side) and $f_\mathrm{sa}=-(\chi-\Delta\chi)\delta z$ (the single-arrow side), respectively, where $\delta z$  denotes the $z$ component of the stretched length of the spring, and $(\chi\pm\Delta\chi)$ denotes the effective stiffnesses for the two sides. Based on the connectivity of the springs in Fig.~\ref{Fig-3Dlattice}(a), the oscillatory equations of the lattice read
\begin{gather}
\begin{aligned}
    m\frac{d^2 }{dt^2}u_A(\mathbf{R}) = & -\kappa_0 u_A(\mathbf{R})+\kappa_1\big[u_B(\mathbf{R})-u_A(\mathbf{R})\big]+\kappa_2\big[u_B(\mathbf{R}-\hat{\mathbf{z}})-u_A(\mathbf{R})\big]\\
    &+\sum_{s_x,s_y=\pm1}(\chi-s_xs_y\Delta\chi)\big[u_B(\mathbf{R}+s_x\hat{\mathbf{x}}+s_y\hat{\mathbf{y}})-u_A(\mathbf{R})\big]+\gamma \frac{d}{dt}u_A(\mathbf{R}),
\end{aligned}\\
\begin{aligned}
    m\frac{d^2 }{dt^2}u_B(\mathbf{R}) = & -\kappa_0 u_B(\mathbf{R})+\kappa_1\big[u_A(\mathbf{R})-u_B(\mathbf{R})\big]+\kappa_2\big[u_A(\mathbf{R}+\hat{\mathbf{z}})-u_B(\mathbf{R})\big]\\
    &+\sum_{s_x,s_y=\pm1}(\chi+s_xs_y\Delta\chi)\big[u_A(\mathbf{R}+s_x\hat{\mathbf{x}}+s_y\hat{\mathbf{y}})-u_B(\mathbf{R})\big]-\gamma \frac{d}{dt}u_B(\mathbf{R}),
\end{aligned}
\end{gather}
where $u_{A/B}(\mathbf{R})$ denotes the $z$-displacement of the oscillator A/B in the unit cell at $\mathbf{R}$. Substituting the Bloch condition $(u_A(\mathbf{R}),u_B(\mathbf{R}))^\intercal=\ket{\psi(\mathbf{k})} e^{i\mathbf{k}\cdot\mathbf{R}-i\omega t}$ into the oscillatory equations, we obtain the quadratic eigenvalue equation for the Bloch modes $\ket{\psi(\mathbf{k})}$ in the momentum space:

\begin{align}
    \Bigg[
&-
\underbrace{\begin{pmatrix}
    (\kappa_0+\kappa_1+\kappa_2+4\chi) & -\kappa_1-\kappa_2 e^{-ik_z}-4(\chi \cos k_x\cos k_y+\Delta\chi\sin k_x\sin k_y) \cr
    -\kappa_1-\kappa_2 e^{ik_z}-4(\chi\cos k_x\cos k_y-\Delta\chi\sin k_x\sin k_y) & (\kappa_0+\kappa_1+\kappa_2+4\chi)\end{pmatrix}
}_{\displaystyle\vb{K}=(\kappa_0+\kappa_1+\kappa_2+4\chi)\sigma_0-(\kappa_1+\kappa_2\cos k_z+4\chi \cos k_x\cos k_y)\sigma_x-(\kappa_2\sin k_z+4i\Delta\chi\sin k_x\sin k_y)\sigma_y}\nonumber\\
&\qquad\qquad+m\omega^2\sigma_0 -i\omega\underbrace{\begin{pmatrix}
    -\gamma & 0 \cr 0 & \gamma
\end{pmatrix}}_{\displaystyle\vb{\Gamma}=-\gamma\sigma_z}
\Bigg]\ket{\psi(\mathbf{k})}=\left[\mathbf{K}(\mathbf{k})+ m \omega^2\sigma_0 -i\omega\mathbf{\Gamma} \right]\ket{\psi(\mathbf{k})}=Q(\omega,\mathbf{k})\ket{\psi(\mathbf{k})}=0.
\end{align}

The mechanical lattice possesses three crystalline symmetries, i.e., $C_{2x}T$, $C_{2y}T$ symmetries, and a generalized $z$-mirror-adjoint symmetry~\cite{zhang2022symmetry} that lead to the formation of an EC. The three symmetries demand the quadratic matrix polynomial $Q(\omega,\mathbf{k})$ to respect the following relations:
\begin{align}
    C_{2x}T\ \text{symmetry}:&\quad \sigma_xQ^*(\omega,-k_x,k_y,k_z)\sigma_x=Q(\omega^*,k_x,k_y,k_z),\label{C2xT}\\
    C_{2y}T\ \text{symmetry}:&\quad \sigma_xQ^*(\omega,k_x,-k_y,k_z)\sigma_x=Q(\omega^*,k_x,k_y,k_z),\label{C2yT}\\
    \text{Generalized } M_z\hbox{-}\dagger\ \text{symmetry}:&\quad \sigma_x Q^\dagger(\omega,k_x,k_y,-k_z)\sigma_x=Q(\omega^*,k_x,k_y,k_z).\label{Mzdagger}
\end{align}
The relations~(\ref{C2xT},\ref{C2yT}) for $C_{2x}T$ and $C_{2y}T$ symmetries exhibit a similar format to the $\gamma$-symmetry of the synthetic dimension model in the main text, and the generalized $M_z\hbox{-}\dagger$ symmetry Eq.~\eqref{Mzdagger} can be directly mapped to the $\kappa$-symmetry of the synthetic dimension model. In particular, the generalized $M_z\hbox{-}\dagger$ symmetry~\eqref{Mzdagger} of the mechanical lattice \textbf{cannot} be expressed as a usual mirror-adjoint symmetry~\cite{zhang2022symmetry} for the linearized Hamiltonian~\eqref{Hamiltonian}, such as $M_z\mathcal{H}^\dagger(k_x,k_y,-k_z)M_z^{-1}=\mathcal{H}(k_x,k_y,k_z)$. Instead, it also serves as a non-Hermitian latent crystalline symmetry for the linearized Hamiltonian, satisfying
\begin{equation}
    \sigma_x\qty(\mathcal{H}(k_x,k_y,-k_z)^n)_{BR}\sigma_x=\qty(\mathcal{H}(k_x,k_y,k_z)^n)^\dagger_{BR},\qquad \forall\ n\in\mathbb{N}.
\end{equation}

As shown in Fig.~\ref{Fig-3Dlattice}(b), the three symmetries~(\ref{C2xT}-\ref{Mzdagger}) confine the ELs in the momentum space into the symmetry-invariant planes of $k_{x,y,z}=0,\pi$. Furthermore, the source-free principle forces the PF ELs in different planes to connect together on the intersecting lines of the planes, resulting in an EC network that spans the entire Brillouin zone.

\subsection{Wavepacket splitting and needle pulses generation at EC point}

\begin{figure*}[t]
\includegraphics[width=0.9\textwidth]{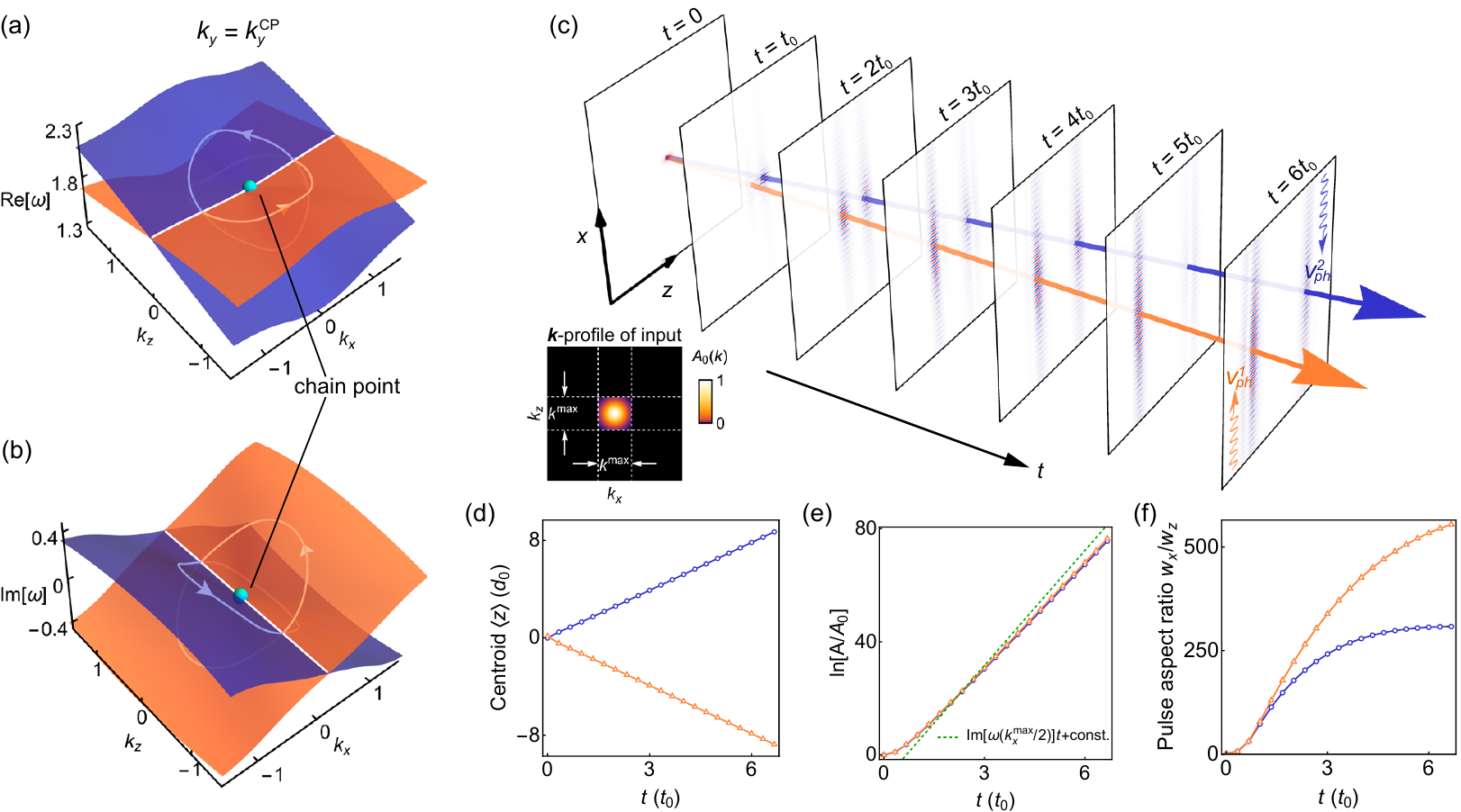}
\caption{\label{Fig-needle pulse}
(a,b) Band structure on the transverse section $k_y=k_y^\mathrm{CP}$ crossing the EC point in Fig.~\ref{Fig-3Dlattice}(c). (c) The time evolution of a wavepacket that is initially localized at a point. For clarity of presentation, the fields on different time slices are normalized independently. The inset in the bottom left illustrates the initial Bloch spectrum $A(\mathbf{k})$ on the $k_y=k_y^\mathrm{CP}$ plane inside the momentum space. The white dashed lines delineate the cutoff $\mathbf{k}$-boundaries outside which the Fourier amplitude rigorously jumps to zero. (d-f) The (d) $z$-coordinate of centroids $\langle{z}\rangle$, (b) the exponential growth rate of field amplitude $\ln(A/A_0)$, and (f) the aspect ratio $w_x/w_z$ of the two splitting pulses as functions of time.
Parameters: $k^\mathrm{max}=0.4\pi$, $q=0.05\pi$, $d_0=1/q$, $t_0=47\mathrm{s}$.}
\end{figure*}

Here, we focus on the bands near the EC point that is marked by the cyan dot in Fig.~\ref{Fig-3Dlattice}(b) and shown in closer detail in Fig.~\ref{Fig-3Dlattice}(c). The coordinates of the EC point can be analytically identified as 
\begin{equation}
    \mathbf{k}^\mathrm{CP}=\qty{0,k_y^\mathrm{CP}=\arccos\qty[\frac{-(\kappa_1+\kappa_2)+\gamma\sqrt{-\gamma^2/4+(\kappa_0+4\chi+\kappa_1+\kappa_2)})}{4\chi}],0}.
\end{equation}
At the chain point, two ingoing red half-ELs and two outgoing blue half-ELs osculate with each other, which can be seen from the twofold braiding of the eigenfrequencies along the loop $S$ as shown in Fig.~\ref{Fig-3Dlattice}(d). This local morphology reproduces the scenario of the synthetic dimension model in Fig.~1(b) of the main text. 

Despite the local configuration of the EC near the chain point bearing resemblance to that of a Hermitian nodal chain, the band dispersion near the chain point is notably different from its Hermitian counterpart. Figure~\ref{Fig-needle pulse}(a,b) depicts the eigenfrequency dispersions on the section plane $k_y=k_y^\mathrm{CP}$ that crosses the chain point, which shows that the real and imaginary parts of the two bands both linearly intersect each other along the lines $k_z=0$ and $k_x=0$, respectively (this is protected by the $C_{2x}T$ and the generalized $M_z\hbox{-}\dagger$ symmetries). Such distinct band dispersion ensures that this non-Hermitian degenerate point is a chain point of two ELs, rather than a general EP, due to the fact that the two complex eigenvalues braid twice along a loop (gray lines on the bands) encircling this point. The linear crossing of the real bands indicates the opposite group velocities of the modes on the two bands near the chain point, notwithstanding that the eigenstates of the two bands coalesce to an identical one at the chain point. And the linear dispersion of the imaginary parts indicates that the signal waveform would undergo peculiar deformations over time when scattered at the chain point.

To describe the transport phenomenon near the EC point, we consider a 2D wavepacket $\Psi(t;x,z)e^{ik_y^\mathrm{CP}y}$ that consists of the Bloch components primarily concentrated in the vicinity of the chain point in the momentum space. The time evolution of this wavepacket can be generally expressed by
\begin{equation}\label{wavepacket}
    \Psi(t;x,z)=\sum_{n=1,2}\Psi_n(t;x,z)=\sum_{n=1,2}\int_{-\pi}^{\pi} dk_xdk_z\,\underbrace{\exp\qty[-i\omega_n(k_x,k_y^\mathrm{CP},k_z)t]\,A_0(k_x,k_z)}_{\displaystyle =A_n(t;k_x,k_z)} \ket{\Psi^R_n(k_x,k_y^\mathrm{CP},k_z)}.
\end{equation}
Here, the 4D vector $\Psi=(u_A,u_B,\frac{d}{dt}u_A,\frac{d}{dt}u_B)^\intercal$ is composed of both the displacements and velocities of the two oscillators in each unit cell labeled by the coordinates $(x,z)$; $\Psi_n(t;x,z)$ ($n=1,2$) denotes the fields contributed from the two bands, respectively, with $\ket{\Psi^R_n(\mathbf{k})}$ and $\omega_n(\mathbf{k})$ being the right eigenvectors and eigenfrequencies of the two positive-frequency eigenstates at $\mathbf{k}$; and $A_n(t;k_x,k_z)$ determines the $\mathbf{k}$-space mode profiles on the two bands changing with time. In addition, we have assumed that the initial mode profiles on the two bands have an identical Gaussian distribution, which is centered at the chain point and truncated by a rectangular window function (see the bottom left inset in Fig.~\ref{Fig-needle pulse}(c)):
\begin{equation}
    A_1(0,k_x,k_z)=A_2(0,k_x,k_z)=A_0(k_x,k_z)=\mathrm{rect}(k_x/k^\mathrm{max},k_z/k^\mathrm{max})\exp[\frac{-({k_x}^2+{k_z}^2)}{4q^2}],
\end{equation}
where $\mathrm{rect}(k_x,k_z)$ denotes the rectangular window function satisfying $\mathrm{rect}(k_x,k_z)=1$ as $|k_x|\leq0.5$ and $|k_z|\leq0.5$, and $\mathrm{rect}(k_x,k_z)=0$ as $|k_x|>0.5$ or $|k_z|>0.5$. The window function is employed for the following reason. The imaginary part $\Im(\omega_n)$ of the eigenfrequencies determines the exponential growth or decay of the Bloch states over time. After a sufficiently long period of time, the Bloch component with the largest $\Im(\omega_n)$ invariably overwhelms all other components and dominates the behavior of the wave. Therefore, to ensure that the propagation behavior of the wavepacket reflects the characteristic of the EC for a long enough time, it is necessary to truncate the Bloch components far away from the chain point.

In Fig.~\ref{Fig-needle pulse}(c), the field distributions computed at different time slices exhibit that the initially point-like wavepacket splits into two pulses propagating in the opposite directions along the $z$ axis. And more intriguingly, the two split pulses rapidly broaden in the $x$ direction and evolve into two needle-like pulses~\cite{wang2010Engineering,parker2016Longitudinal,grunwald2020Needle,cao2023Opticalresolution} with  tremendous aspect ratios.  This brand-new effect is indeed a unique phenomenon stemming from the concurrent linear intersections of the real and imaginary parts of the bands at the EC point (see Fig.~\ref{Fig-needle pulse}(a,b)). In what follows, we explain the relationship between the unique band structure at the EC point and the formation of the needle pulses. 

First, the splitting of the wavepacket is attributed to the linear crossing of the real parts of the eigenfrequency bands. The centroid positions of the two split pulses $\Psi_{1/2}(t;x,z)$, which are composed of the modes on the two bands, respectively, are depicted in Fig.~\ref{Fig-needle pulse}(d). The linear trajectories demonstrate the opposite group velocities of the two pulses. Second, the wave deformation from point-like pulses to needle pulses is induced by the linear crossing of the imaginary parts of the bands. As shown in Fig.~\ref{Fig-needle pulse}(b), the two imaginary bands $\Im(\omega_{1,2}(k_x,k_y^\mathrm{CP}k_z))$ exhibit linear dispersion in the $k_x$ direction, while they are nearly dispersionless along the $k_z$ axis.  Therefore, the wave components at the opposite edges $k_x=k^\mathrm{max}/2$ and $k_x=-k^\mathrm{max}/2$ of the window function have the highest exponential growth rate $\Im(\omega^\mathrm{max})=\Im(\omega_{1}(k^\mathrm{max}/2,k_y^\mathrm{CP},k_z))=\Im(\omega_{2}(-k^\mathrm{max}/2,k_y^\mathrm{CP},k_z))$ in the orange and purple pulses, respectively. Consequently, these marginal components become the dominant Bloch components over time. The numerical results in Fig.~\ref{Fig-needle pulse}(e) support this analysis by demonstrating that the amplitudes of the two pulses approach a uniform exponential growth with the rate $\Im(\omega^\mathrm{max})$ (see the green dashed line). Thus, as time progresses, the two pulses gradually expand their transverse width, $w_x$, and approach two plane waves along the $x$ direction with opposite phase velocities $v_{1,2}^\mathrm{ph}\approx\pm 2\Re(\omega^\mathrm{CP})/k^\mathrm{max}$ (see the wavy arrows in the last slice in Fig.~\ref{Fig-needle pulse}(c)). Meanwhile, the pulses experience negligible spatial diffraction and maintain small pulse widths $w_z$ in the $z$ direction thanks to the flat dispersion of the imaginary bands in this direction. As shown in Fig.~\ref{Fig-needle pulse}(f), this novel anisotropic wave-shaping mechanism results in the formation of two transverse needle pulses (here, ``transverse'' means that the long axis of the pulses is transverse to the propagation direction of it) with giant aspect ratios, $w_x/w_z$, of up to the order of hundreds.

In summary, our findings reveal that the unusual complex band dispersion of ECs can give rise to the unprecedented wave-shaping effect of splitting a rectangular wavepacket into two transverse needle pulses with perpendicular group and phase velocities that both take opposite directions in the two pulses. We believe the contribution of this discovery is not limited to proposing an alternative method for generating needle pulses but also opens up a new direction for exploring the physical effects and potential applications arising from non-Hermitian EPs.

\hspace{-60pt}
\appendix

\section*{Appendix: Proof of Theorem~\ref{latent symmetry}}\label{appendix A}

\begin{proof}
\textbf{Left $\Leftarrow$ right}. First, we prove ``left $\Leftarrow$ right'' by mathematical induction. For $n=0$, the left equality is trivially satisfied. Assuming $L\left(\mathcal{H}^{m}\right)_{SS} L^{-1}=\left(\mathcal{H}^{n-1}\right)_{SS}^\dagger$ holds for $m\leq n-1$, next we need to prove that it is also satisfied for $n$. To see this, we derived the recurrence relation of $\mHss[n]$ from $\mH^n=\mH^{n-1}\mH=\mH\mH^{n-1}$: 
\begin{equation}
   \begin{split}\label{recursion of Hss}
   \mHss[n]&=\mHss[n-1]\mH_{SS}+\sum_{m=0}^{n-2} \mHss[m]\mH_{S\bar{S}}(\mH_{\bar{S}\bar{S}})^{n-m-2}\mH_{\bar{S}S}\\
   &=\mH_{SS}\mHss[n-1]+\sum_{m=0}^{n-2}\mH_{S\bar{S}}(\mH_{\bar{S}\bar{S}})^{n-m-2}\mH_{\bar{S}S} \mHss[m].
\end{split}
\end{equation}
From $L\mathcal{R}_S(\mathcal{H},\omega)L^{-1}=\mathcal{R}_S(\mathcal{H},\omega^*)^\dagger$, we have \begin{gather}
    L\mathcal{H}_{SS}L^{-1}=\lim_{\omega\rightarrow\infty}L\mathcal{R}_S(\mathcal{H},\omega)L^{-1}=\lim_{\omega\rightarrow\infty}\mathcal{R}_S(\mathcal{H},\omega^*)^\dagger=\mathcal{H}_{SS}^\dagger,\label{eq 1}\\
    L\qty[\mH_{S\bar{S}}(\mH_{\bar{S}\bar{S}}-\omega I)^{-1}\mH_{\bar{S}S}]L^{-1}=\qty[\mH_{S\bar{S}}(\mH_{\bar{S}\bar{S}}-\omega^* I)^{-1}\mH_{\bar{S}S}]^\dagger,\label{eq 2}\\
    L\qty[\mH_{S\bar{S}}\qty(\mH_{\bar{S}\bar{S}})^{n}\mH_{\bar{S}S}]L^{-1}=\qty[\mH_{S\bar{S}}\qty(\mH_{\bar{S}\bar{S}})^{n}\mH_{\bar{S}S}]^\dagger\quad \forall\ n\in\mathbb{N}.\label{eq 3}
\end{gather}
Here, Eq.~\eqref{eq 3} is obtained from Eq.~\eqref{eq 2} and the Neumann series $(\mH_{\bar{S}\bar{S}}-\omega I)^{-1}=\sum_{n=0}^{\infty}\qty(1/\omega)^{n+1}\qty(\mH_{\bar{S}\bar{S}})^n$ for sufficiently large $\omega$ such that $\det[\mH_{\bar{S}\bar{S}}/\omega]<1$ (the convergent condition of the series).

In terms of Eqs.~(\ref{recursion of Hss},\ref{eq 1},\ref{eq 3}) and $L\left(\mathcal{H}^{m}\right)_{SS} L^{-1}=\left(\mathcal{H}^{n-1}\right)_{SS}^\dagger$ for $m\leq n-1$, we acquire
\begin{equation}
\begin{split}    
    L\mHss[n]L^{-1} &=\qty(L\mHss[n-1]L^{-1})\qty(L\mH_{SS}L^{-1})+\sum_{m=0}^{n-2} \qty(L\mHss[m]L^{-1})\qty(L\mH_{S\bar{S}}(\mH_{\bar{S}\bar{S}})^{n-m-2}\mH_{\bar{S}S}L^{-1})\\
    &=\mHss[n-1]^\dagger\mH_{SS}^\dagger+\sum_{m=0}^{n-2}(\mH^m)_{SS}^\dagger\qty(\mH_{S\bar{S}}(\mH_{\bar{S}\bar{S}})^{n-m-2}\mH_{\bar{S}S})^\dagger\\
    &=\qty[\mH_{SS}\mHss[n-1]+\sum_{m=0}^{n-2}\mH_{S\bar{S}}(\mH_{\bar{S}\bar{S}})^{n-m-2}\mH_{\bar{S}S}(\mH^m)_{SS}]^\dagger=(\mH^n)_{SS}^\dagger.
\end{split}
\end{equation}
From the induction, we have proved ``left $\Leftarrow$ right''.

\textbf{Left $\Rightarrow$ right}. From the recurrence relation~\eqref{recursion of Hss} and $L\left(\mathcal{H}^n\right)_{SS} L^{-1}=\left(\mathcal{H}^n\right)_{SS}^\dagger\ \forall\  n\in\mathbb{N}$, we can also prove the following relation by induction:
\begin{equation}
    L\qty[\mH_{S\bar{S}}\qty(\mH_{\bar{S}\bar{S}})^{n}\mH_{\bar{S}S}]L^{-1}=\qty[\mH_{S\bar{S}}\qty(\mH_{\bar{S}\bar{S}})^{n}\mH_{\bar{S}S}]^\dagger\quad \forall\ n\in\mathbb{N}.\label{eq 4}
\end{equation}
On the other hand, by the Cayley-Hamilton theorem, the inverse of a matrix $A$ can be expressed as $A^{-1}=-\sum_{k=0}^{\mathrm{dim} A}\frac{c_k}{c_0}A^k$ where $c_k$ ($k=0,\cdots,\mathrm{dim} A$) are the coefficients of the characteristic polynomial of $A$: $p_A(\lambda)=\det[A-\lambda]=\sum_{k=0}^{\mathrm{dim} A} c_k \lambda^k$ with $c_0=(-1)^{\mathrm{dim} A}\det[A]\neq0$. Using this formula, the isospectral reduction can be expanded as
\begin{equation}
\begin{split}\label{expansion of R}   
    \mathcal{R}_S(\mathcal{H},\omega)&=\mathcal{H}_{SS}-\mH_{S\bar{S}}\sum_{k=0}^{\mathrm{dim} \bar{S}}\frac{c_k(\omega)}{c_0(\omega)}(\mH_{\bar{S}\bar{S}}-\omega I)^k\mH_{\bar{S}S}\\
    &=\mathcal{H}_{SS}-\sum_{k=0}^{\mathrm{dim} \bar{S}}\frac{c_k(\omega)}{c_0(\omega)}\sum_{n=0}^k C_k^n(-\omega)^{k-n} \mH_{S\bar{S}}(\mH_{\bar{S}\bar{S}})^n\mH_{\bar{S}S},
\end{split}
\end{equation}
where $c_k(\omega)$ are the coefficients of the characteristic polynomial of $(\mH_{\bar{S}\bar{S}}-\omega I)$, which satisfy $c_k(\omega^*)=c_k(\omega)^*$ due to $\det[\mH_{\bar{S}\bar{S}}-(\omega+\lambda) I]^*=\det[\mH_{\bar{S}\bar{S}}^\dagger-(\omega^*+\lambda^*) I]$, and $C_k^n=\frac{k!}{n!(k-n)!}$ are the binomial coefficients. Using Eq.~\eqref{eq 4} and \eqref{expansion of R},
\begin{equation}
\begin{split}    
    L\mathcal{R}_S(\mathcal{H},\omega)L^{-1}&=L\mathcal{H}_{SS}L^{-1}-\sum_{k=0}^{\mathrm{dim} \bar{S}}\frac{c_k(\omega)}{c_0(\omega)}\sum_{n=0}^k C_k^n(-\omega)^{k-n} L\qty[\mH_{S\bar{S}}(\mH_{\bar{S}\bar{S}})^n\mH_{\bar{S}S}]L^{-1}\\
    &=\mathcal{H}_{SS}^\dagger-\sum_{k=0}^{\mathrm{dim} \bar{S}}\frac{c_k(\omega^*)^*}{c_0(\omega^*)^*}\sum_{n=0}^k C_k^n(-\omega)^{k-n} \qty[\mH_{S\bar{S}}(\mH_{\bar{S}\bar{S}})^n\mH_{\bar{S}S}]^\dagger
    =\mathcal{R}_S(\mathcal{H},\omega^*)^\dagger,
\end{split}
\end{equation}
which completes the proof.
\end{proof}

\bibliography{reference.bib}